\documentclass[a4paper, UKenglish, cleveref, autoref]{lipics-v2021}
\usepackage{multicol}
\usepackage{amsmath}
\usepackage{mathtools}
\usepackage{tikz}
\usetikzlibrary{arrows.meta,decorations.markings,decorations.pathmorphing,calc,backgrounds,shapes.misc,shapes.geometric,bending}
\usepackage{xcolor}
\def\loopangle{30}
\def\looplooseness{15}
\tikzset{
    >={Stealth[length=2.25mm]},
    deviation/.style={black!50},
    connexion/.style={orange!75!black},
    vertex/.style={fill,black,inner sep=0pt,minimum size=4pt,outer sep=2pt, circle},
    edge/.style={->, black, draw, very thick,},
    node/.style={rounded rectangle, draw, black,inner sep=2pt,minimum size=4pt,outer sep=2pt},
    tape/.style={rounded rectangle, draw=black,inner sep=2pt,minimum size=4pt,outer sep=2pt},
    double tape/.style={rounded rectangle, draw=black,inner sep=1pt,minimum size=4pt,outer sep=1pt, double distance=1pt},
    vm loop/.style={pos=.5, looseness = \looplooseness},
    north west loop/.style={vm loop, in={\the\numexpr 135 + \loopangle\relax}, 
                                     out ={\the\numexpr 135 - \loopangle\relax}},
    north east loop/.style={vm loop, in={\the\numexpr 45 + \loopangle\relax}, 
                                     out ={\the\numexpr 45 - \loopangle\relax}},
    south west loop/.style={vm loop, in={\the\numexpr -135 + \loopangle\relax}, 
                                     out ={\the\numexpr -135 - \loopangle\relax}},
    south east loop/.style={vm loop, in={\the\numexpr -45 + \loopangle\relax}, 
                                     out ={\the\numexpr -45 - \loopangle\relax}},
    north loop/.style={vm loop, in={\the\numexpr 90 + \loopangle\relax}, 
                                out ={\the\numexpr 90 - \loopangle\relax}},
    south loop/.style={vm loop, in={\the\numexpr 270 - \loopangle\relax}, 
                                out ={\the\numexpr 270 + \loopangle\relax}},
    east loop/.style={vm loop, in={\the\numexpr 0 + \loopangle\relax}, 
                                out ={\the\numexpr 0 - \loopangle\relax}},
    west loop/.style={vm loop, in={\the\numexpr 180 - \loopangle\relax}, 
                                out ={\the\numexpr 180 + \loopangle\relax}},
    /tikz/on layer/.code={
        \pgfonlayer{#1}\begingroup
        \aftergroup\endpgfonlayer
        \aftergroup\endgroup
    }, 
    order/.style={cbdarkblue,font=\footnotesize, inner sep=0pt, outer sep=0pt},
}

\definecolor{cborange}{HTML}{d55e00}
\definecolor{cbpink}{HTML}{cc79a7}
\definecolor{cbblue}{HTML}{0072b2}
\definecolor{cbyellow}{HTML}{f0e442}
\definecolor{cbteal}{HTML}{009e73}
\colorlet{cbdarkteal}{cbteal!75!black}
\colorlet{cbdarkorange}{cborange!90!black}
\colorlet{cbdarkblue}{cbblue!90!black}
\colorlet{cbdarkpink}{cbpink!50!black}
\usepackage[super]{nth}
\def\nthscript#1{\ifmmode^{\text{#1}}\else$^{\text{#1}}$\fi}
\usepackage[bibliography=common]{apxproof}

\usepackage{adjustbox}

\newcommand{\last}{\textsf{last}}
\newcommand{\sspecial}{\textsf{special}}

\newcommand{\calM}{\mathcal{M}}
\newcommand{\calD}{\mathcal{D}}
\newcommand{\calT}{\mathcal{T}}

\newcommand{\NN}{\mathbb{N}}
\newcommand{\ZZ}{\mathbb{Z}}
\newcommand{\bool}{\{0,1\}}

\newcommand\mirror{\mathrm{Mir}}

\AddToHook{env/proposition/begin}{\crefalias{theorem}{proposition}}
\AddToHook{env/lemma/begin}{\crefalias{theorem}{lemma}}
\AddToHook{env/example/begin}{\crefalias{theorem}{example}}
\AddToHook{env/definition/begin}{\crefalias{theorem}{definition}}
\AddToHook{env/corollary/begin}{\crefalias{theorem}{corollary}}
\AddToHook{env/claim/begin}{\crefalias{theorem}{claim}}
\AddToHook{env/remark/begin}{\crefalias{theorem}{remark}}

\newcommand{\upstate}{\uparrow}%
\newcommand{\downstate}{\downarrow}%
\colorlet{darkorange}{orange!75!black}
\colorlet{darkblue}{blue!75!black}
\newcommand{\omark}[2][draw=red, line width=1pt]{\tikz[baseline=(tmp.base)]{%
    \node[inner sep=0pt](tmp){$#2$}; 
    \path (tmp.base) ++ (0,-.2em) coordinate (below);
    
    \path[#1] 
        ($(below)+(-.1em,-.2em)$) -- (below) -- ($(below)+(.1em,-.2em)$) ;
    \pgfresetboundingbox
    \useasboundingbox (tmp.north east) rectangle (tmp.south west);
}}
\newcommand{\dimmedomark}[1]{\omark[draw=black!50, line width=1pt]{#1}}

\newcommand{\OO}{\mathrm{O}}

\newcommand{\LL}{\mathrm{L}}
\newcommand{\RR}{\mathrm{R}}
\newcommand{\bow}{{\triangleright}}
\newcommand{\eow}{{\triangleleft}}
\newcommand{\Even}{\mathsf{Even}}
\newcommand{\Odd}{\mathsf{Odd}}

\newcommand{\ii}{\mathrm{i}}
\newcommand{\ff}{\mathrm{h}}

\renewcommand\leadsto[1][]{%
{\begin{tikzpicture}%
\path coordinate (START) 
    ++(2.5ex,0) coordinate (MID)
    ++(1.1ex,0) coordinate (END);
\path[line width=.6pt,draw=white,-{Stealth[black,length=1.1ex,width=1.1ex]}] (MID) -- (END);
\draw[
    line width=.6pt,
    cap=round,
    decorate, decoration={
        snake, segment length=1ex, amplitude=.2ex, 
        pre length=0pt, post length=0cm},
    ] (START) -- node[midway,above]{\scriptsize$#1$}(MID);
\end{tikzpicture}\hskip-.1ex}%
}%

\DeclareRobustCommand
  \thincdots{\mathinner{\raisebox{.1ex}{$\cdotp\mkern-2mu\cdotp\mkern-2mu\cdotp$}}}

\DeclarePairedDelimiter\sparenv{\lbrack}{\rbrack}

\DeclareRobustCommand{\svdots}{%
  \vcenter{%
    \offinterlineskip
    \hbox{.}
    \vskip0.25\normalbaselineskip
    \hbox{.}
    \vskip0.25\normalbaselineskip
    \hbox{.}%
    \vskip0.125\normalbaselineskip
  }%
}

\newcommand{\vmtopstrut}{\rule{0pt}{1.1em}}
\newcommand{\vmbotstrut}{\rule[-.55em]{0pt}{1em}}

\newcommand{\qi}{q_{\ii}}
\newcommand{\qh}{q_{\ff}}

\newcommand{\hlpath}[1]{\textcolor{cbdarkorange}{#1}}
\newcommand{\hlpari}[1]{\textcolor{cbdarkteal}{#1}}

\newcommand{\ourimplemurl}{https://gitlab.com/marsault/constant-delay-gray-code-public/-/releases/v1.0}

\newcommand{\parcode}[1]{#1}

\title{Gray Codes with Constant Delay\texorpdfstring{\\}{} and Constant Auxiliary Space}

\titlerunning{Gray Codes With Constant Delay and Constant Auxiliary Space} %

\author{Antoine Amarilli}
       {{Univ.~Lille, Inria, CNRS, Centrale Lille, UMR 9189 CRIStAL, F-59000 Lille, France}}
       {antoine.a.amarilli@inria.fr}
       {https://orcid.org/0000-0002-7977-4441}
       {}
\author{Claire David}
       {Univ.~Gustave Eiffel, CNRS, LIGM, F-77454 Marne-la-Vallée, France}
       {claire.david@univ-eiffel.fr}
       {https://orcid.org/0000-0002-3041-7013}
       {}
\author{Nadime Francis}
       {Univ.~Gustave Eiffel, CNRS, LIGM, F-77454 Marne-la-Vallée, France}
       {nadime.francis@univ-eiffel.fr}
       {https://orcid.org/0009-0009-4531-7435}{}
\author{Victor Marsault}
       {Univ.~Gustave Eiffel, CNRS, LIGM, F-77454 Marne-la-Vallée, France
        \and \url{https://victor.marsault.xyz}}
       {victor.marsault@univ-eiffel.fr}
       {https://orcid.org/0000-0002-2325-6004}
       {}
\author{Mikaël Monet}
       {Univ.~Lille, Inria, CNRS, Centrale Lille, UMR 9189 CRIStAL, F-59000 Lille, France}
       {mikael.monet@inria.fr}
       {https://orcid.org/0000-0002-6158-4607}
       {}
\author{Yann Strozecki}
       {Université de Versailles Saint-Quentin, Laboratoire David, France}
       {yann.strozecki@uvsq.fr}
       {https://orcid.org/0000-0002-0891-3766}
       {}

\authorrunning{A. Amarilli, C. David, N. Francis, V. Marsault, M. Monet, and Y. Strozecki} %

\Copyright{Antoine Amarilli, Claire David, Nadime Francis, Victor Marsault, Mikaël Monet, and Yann Strozecki} %

\relatedversiondetails[cite=CDCASGrayCode-icalp,linktext={Appeared in ICALP'26}]{Conference proceeding}{https://doi.org/10.4230/LIPIcs.ICALP.2026.160} %

\supplementdetails[subcategory={Gitlab Repository},linktext={Tape and Deque Machines Implementing Gray codes}, cite=OurImplementation]{Prototype implementation}
{\ourimplemurl} %



\acknowledgements{%
We are grateful to Gabriel Bathie and Florent Capelli for their insights in discussions about the problem,
and Quentin Bartolone for spotting a few errors in the tabular definitions of machines.
We also thank Thomas Colcombet and IRIF for their support.
}

\newtheoremrep{theorem}{Theorem}[section]       
\newtheoremrep{proposition}{Proposition}[section]
\newtheoremrep{lemma}{Lemma}[section]
\newtheoremrep{claim}{Claim}[section]

\hideLIPIcs
\nolinenumbers

\begin{document}

\maketitle

\begin{abstract}
    We give the first two algorithms to enumerate all binary words of
    $\{0,1\}^\ell$ (like Gray codes) while ensuring that the delay and the auxiliary space
    is independent from~$\ell$, i.e., constant time for each word, and constant
    memory in addition to the $\ell$ bits storing the current word.
    Our algorithms are given in two new computational models: \emph{tape machines} and \emph{deque machines}. 
    We also study more restricted models, \emph{queue machines} and \emph{stack machines}, and show that they cannot enumerate all binary words with constant auxiliary space, even with unrestricted delay.
    
    A \emph{tape machine} is a Turing machine that stores the current binary
    word on a single working tape of length $\ell$ (which never increases),
    using no other tape.
    The machine has a single head and must edit its tape to reach all possible words of $\{0,1\}^\ell$, and output them (in unit time, by entering special output states), with no duplicates. 
    Hence a tape machine uses constant auxiliary space by definition (up to the
    head position).
    We construct a tape machine that achieves this task with constant delay
    between consecutive outputs, so that the machine implements a so-called
    \emph{skew-tolerant quasi-Gray code}. We then construct a more involved tape machine that implements a Gray code.

    A \emph{deque machine} stores the current binary word on a double-ended queue of length~$\ell$, and stores a constant-size internal state. 
    It works as a tape machine, except that it modifies the content of the deque by performing push and pop operations on the endpoints.
    Hence again a deque machine uses constant auxiliary space by definition.
    We construct deque machines that enumerate all words of $\{0,1\}^\ell$ with constant-delay.  The main technical challenge in this model is to correctly detect when enumeration has finished. 
\end{abstract}

\begin{toappendix}
\section{Implementation}
\label{app:implementation}
In addition to this paper, we wrote several files and scripts~\cite{OurImplementation}
that can be downloaded from:\\[1mm]
\hspace*{0cm}\hfil\adjustbox{max width=\linewidth}{\url{\ourimplemurl}}\\[1mm]
These files provide implementations of our tape and deque machines, along with scripts that verify with empirical checks that they meet our claims.
Here is a short description of the files provided and their usage.

\begin{itemize}\setlength\itemsep{\medskipamount}
    \item The files \texttt{T0.txt}, \texttt{T1.txt}, \texttt{T2.txt}, \texttt{D0.txt}, \texttt{D1.txt}
    and \texttt{D2.txt} contain the transitions table of the machines~$\calT_0$ (\cref{fig:tape-easy-table}), $\calT_1$ (\cref{fig:tape-ham1-table}), $\calT_2$ (\cref{fig:tape-hs24-table}), $\calD_0$ (\cref{fig:deque-easy-table}),~$\calD_1$ (\cref{fig:deque-double-table}),~$\calD_2$ (\cref{fig:deque-lookahead-table}), respectively.
    
    Note that these are data files: they must be launched with scripts \texttt{tapeMachine.py} or \texttt{dequeMachine.py}

    \item The Python script \texttt{tapeMachine.py} runs a tape machine given by its table. 
    It has several options, in particular \texttt{-{}-word\_size} to specify the size of the words produced.
    Example usage:
    \begin{itemize}
        \item \texttt{python3 tapeMachine.py T0.txt -{}-word\_size 5} prints all words of size $5$
        \item \texttt{python3 tapeMachine.py T1.txt  -{}-verbose} prints all words of size 6 (default size) and output the sequence of configurations of the machine
        \item \texttt{python3 tapeMachine.py T1.txt -{}-reverse} runs the machine, whose transitions are the transitions of \texttt{T1.txt} reversed
    \end{itemize}

    \item The Python script \texttt{dequeMachine.py} runs a deque machine given by its table, with the same options and 
    usage as \texttt{tapeMachine.py}.
    
    \item The script \texttt{check.sh} tests whether all the machines are Hamiltonian (or prefix-Hamiltonian, for $\calD_0$), and their delay is checked against a given bound, for sizes 3 to 18. It also checks the Hamming-1 condition on $\calT_1$, and the bound on Hamming distance for $\calT_0$.
    
    \item The Python script \texttt{rankT1.py} computes the rank of any word of length~$\ell$ in the code for length~$\ell$ produced by the tape machine~$\calT_1$ (see
    Appendix~\ref{apx:rankunrank}).

    For instance, command \texttt{python3 rankT1.py 01001} prints \texttt{6}, which is correct according to \cref{fig:tape-ham1-tree}: 01001 is the \nth{7} word to be enumerated.
    
    \item The Python script \texttt{unrankT1.py} takes as input a rank~$n$ and a word length~$\ell$ and prints the word at rank~$n$ (that is, the $(n+1)$-th word) in the code for length~$\ell$ implemented by the tape machine~$\calT_1$ (see
    Appendix~\ref{apx:rankunrank}).

    For instance, command \texttt{python3 unrankT1.py 22 5} prints \texttt{10010}, which is correct according to \cref{fig:tape-ham1-tree}: 10010 is the \nth{23} word to be enumerated.

    \item The Python script \texttt{checkRankAndUnrankT1.py} checks that the three scripts \texttt{rankT1.py}, \texttt{unrankT1.py} and \texttt{tapeMachine.py T1.txt} are consistent:
    it launches \texttt{tapeMachine.py} as a subprocess, and, for each word output by \texttt{tapeMachine.py}, it launches \texttt{unrankT1.py} and \texttt{rankT1.py} to check they output the correct result.

    For instance, \texttt{python3 checkRankAndUnrankT1.py 8} checks the three scripts are consistent for word-length~$8$.

\end{itemize}

\end{toappendix}

\section{Introduction}
In this paper we propose new efficient
algorithms to produce all binary words of a certain
length $\ell \in \NN$, i.e., all words of $\bool^\ell$.
Specifically, we are looking for efficient
\emph{enumeration algorithms}, 
which produce these words one after another, without
repetitions, and while minimizing the
\emph{delay}
between two consecutively produced words.
If we require each word to be written from scratch,
then obviously we cannot hope for a better delay than $O(\ell)$.
For this reason, we adopt the “Do not count the
output” principle from the enumeration algorithms
literature~\cite[p.~8]{ruskey2003combinatorial}, in the spirit of 
\emph{combinatorial Gray codes}
\cite{savage1997survey,Mutze23-gray}.
In this setting, a single word is explicitly maintained in memory,
the algorithm has an $\mathsf{output}$ instruction that produces the current
word in unit time,
and the word can be modified in-place with a certain
set of edit operations, before calling the next
$\mathsf{output}$ instruction.
Of course, what can be achieved depends on which update operations are allowed.
Two common choices are \emph{substitutions} (i.e., store the current word in an
array and make it possible to change the $i$-th symbol) and \emph{endpoint
modifications} (i.e., store the current word in a doubly-linked list and make it
possible to push or pop on the endpoints).

In this model, for each set of allowed edit operations, we can try to optimize
two things: running time, and memory usage. For running time, the best possible
goal is to achieve \emph{constant delay}, i.e., the number of computation steps
and edits from one word produced to the next is bounded by a constant (independent from
the word length $\ell$). For memory usage, the best possible goal is to achieve
\emph{constant additional memory}, i.e., the memory state consists of the
currently produced word (i.e., $\ell$ bits) plus a finite state which is again
independent from the word length $\ell$.

For the more commonly studied setting of substitution updates, 
it is already known that each of these two aspects can be optimized in isolation.
We can achieve constant delay using, for instance, the 
\emph{Binary Reflected Gray Code} (BRGC)~\cite{rahman2010integer}:
start from the word $0^\ell$,
flip the rightmost bit on odd steps, and
flip the bit that is to the left of the rightmost~$1$ on even steps.
This gives what is known in this line of work as a \emph{loopless} or \emph{loop-free} algorithm~\cite{ehrlich1973loopless}.
However, this algorithm requires non-constant additional memory: to quickly find the bit
to flip on even steps, we would typically use a linear-sized stack
storing the current positions of the $1$'s;
see also~\cite{bitner1976efficient,rahman2010integer}.
As for \emph{constant additional memory}, the optimal memory required to
implement Gray codes has been studied in the \emph{decision assignment tree}
(DAT) model or bit-probe
model~\cite{fredman1978observations,bose2010improved,brodal2014integer,rahman2010integer,raskin2017linear},
with essentially a logarithmic lower bound.
However, these models are \emph{non-uniform}: they assume that the algorithm can
freely access the entire word and that it can perform arbitrary computations
with possibly high running time as a function of the word length.

For the setting of endpoint modifications, the situation is far more open.
There are known related constructions to enumerate binary words, e.g.,
\emph{universal words}~\cite{chung1992universal} or Hamiltonian paths in
so-called \emph{Shuffle Exchange networks}~\cite{feldmann1996shuffle}, but to
our knowledge they do not give favorable bounds on the delay and memory usage
of an algorithm. We also note that an efficient algorithm has
recently been proposed for the so-called \emph{Sigma-Tau
problem}~\cite{sawada2019solving,liptak2023constant}. This problem consists in
enumerating all permutations of $\{1,\ldots,n\}$ using specific endpoint
modification operations, and as far as we can tell this is a different problem
from ours.

Hence, our focus in this work is whether we can achieve the two requirements
simultaneously: Can we enumerate $\{0,1\}^\ell$ with an algorithm that uses
only constant auxiliary memory and produces each word in constant delay?
Our results show that the answer
depends on the set of edit operations that are allowed, and
on the computational model that is used (e.g., 
which instructions are allowed, and
what counts as memory).
We study this problem first for a kind of Turing machines
called \emph{tape machines}, and then for
another formalism 
that we introduce, called \emph{deque machines}.
In these two frameworks, we construct machines that 
enumerate $\{0,1\}^\ell$
with constant delay and while using constant auxiliary space.
Note that this does not contradict the known lower bounds from the DAT model: 
we discuss this in more detail in \cref{sec:tape} for tape machines, and for deque
machines we are not aware of known lower bounds because deque machines can 
perform circular shifts in unit time.
We present our contributions in more detail below.

\subparagraph{Contributions.}
We first present \emph{tape machines}, which are a kind of deterministic Turing machines that 
also resemble \emph{linear bounded automata}~\cite{kuroda1964classes}.
Tape machines 
have
a finite set of states, some of which are distinguished as output states,
along with a single working tape and one head. The tape alphabet is $\{0,1\}$. The tape is initialized with the word $0^\ell$ and
stays of length~$\ell$ throughout.
Like a Turing machine, a tape machine updates its tape based on its internal
state and on the symbol under the head.
It may then change the bit under the head, move its head left or right, and change its state.
Note that the machine only has a single working tape, and no input or output tape.
The word on the working tape is produced in unit time
whenever an output state is reached.

This definition ensures that tape machines satisfy our constant-memory requirement:
in addition to the working tape used to store the current word, 
the only auxiliary memory that they use is their internal state,
from a constant-sized set of states which is independent from~$\ell$.
Our goal is to design a \emph{Hamiltonian} tape machine:
for every word length~$\ell\geq 1$, when started on the word $0^\ell$, the machine
must produce all the words of $\{0,1\}^\ell$ exactly once (i.e., each word is seen exactly once in an
output state), and must then halt.
Further, the machine must be \emph{constant-delay}:
there must be a constant $B$ independent from~$\ell$ such that the machine produces a new word after at most~$B$ steps.
It is not obvious at all from the definition that Hamiltonian constant-delay tape machines exist, because they have no working memory beyond the tape used to store the words to be produced.
In particular, the constant-delay requirement implies that the bits flipped from one word to the next must be at constant distance from one another (i.e., if we flip bit $i$ and then bit $j$ then the distance $|j-i|$ must be small). This latter requirement is not satisfied by the BRGC, but it is obeyed by some Gray codes, for instance, the \emph{$k$-skew-tolerant} codes
of~\cite{sachimelfarb2025improved}
(attributed to~\cite{WilmerErnst2002}), 
where~$k$ is the maximum distance between consecutively flipped bits. 
Our first contribution is to show that we can in fact design a Hamiltonian tape machine. What is more, we can strengthen the result and ensure that the machine is \emph{Hamming-1}, i.e., the Hamming distance between any two successively produced words is 1: this property is satisfied by Gray codes such as the BRGC or the ones 
of~\cite{sachimelfarb2025improved}). 
Formally:

\begin{theorem}
  \label{thm:tape-intro}
There exists a Hamiltonian constant-delay Hamming-1 tape machine.
\end{theorem}

The result is intuitively established by implementing a traversal of a complete
binary tree with the “even-odd
trick”~\cite{karaganis1968cube,sekanina1960ordering}, and refining it to make
it Hamming-1. It turns out that our machine implements
a code closely related to \cite[Construction A]{sachimelfarb2025improved} (with $k=3$), despite having no apparent relationship to the simple inductive definition
of~\cite{sachimelfarb2025improved}.
Thus, we can also see \cref{thm:tape-intro} as a construction showing that (a variant of) the code of \cite[Construction A]{sachimelfarb2025improved} can be implemented to use only constant delay and constant auxiliary memory.
We explain how the two codes relate in Appendix~\ref{app:hs24}.
We also show that our machine can implement counters by detailing the decrement
function, and discuss rank and unrank functions in Appendix~\ref{apx:rankunrank}.
We relate our implementation to works using the DAT model \cite{fredman1978observations,bose2010improved,rahman2010integer}, and show that it is competitive with the best counter implementations in the RAM model~\cite{elmasry2013place}.

Our second contribution is to study a 
computational model 
that we call the \emph{deque machine}, in line with other works studying machines
augmented with deques~\cite{ayers1985deque,petersen2001stacks}.
In this model, the word is stored in a double-ended queue, the machine can
read the first and last bits of the word, and it can perform push and pop operations on the left and right endpoints.
The motivation for this model is to ensure that the entire configuration of the
machine consists of the current word plus the finite state, in contrast with 
tape machines that arguably hide $\lceil \log \ell \rceil$ auxiliary memory 
bits %
in 
the position of the head.
Moreover, deque machines are effective models to implement push-pop edits on words, as studied in \cite{amarilli2023enumerating} for enumeration of regular languages.
Finally, allowing only endpoints modifications is natural in other contexts,
such as enumerating walks in a graph, which is a core task in modern graph
database query processing \cite{DeutschEtAl22}. Indeed, the task of enumerating
binary words studied in the present paper is a first step towards the efficient
enumeration of walks in graphs: it corresponds to walk enumeration for the specific graph that has one
vertex and two self-loops labeled 0 and 1.

It is known that we can indeed enumerate all binary words with constantly many endpoint modifications between each word, as follows from 
\emph{universal words} \cite{chung1992universal} or Hamiltonian paths in so-called \emph{Shuffle
Exchange networks} \cite{feldmann1996shuffle}.
So our second contribution is to design deque machines that can produce such codes while only using constant auxiliary memory:

\begin{theorem}
  \label{thm:deque-intro}
  There exists a Hamiltonian constant-delay deque machine.
\end{theorem}

The first step to show this result is to notice that the tape machine from \cref{thm:tape-intro}
can be adapted 
into a constant delay deque machine that produces all binary words once but then continues to run indefinitely and produces duplicates.
As it turns out, the main challenge is to detect when the enumeration is complete, for which we present two alternative techniques which we hope can be of independent interest.

Our third contribution is to show that the model of deque machine is
``minimal'' in the sense that further restricting the edit operations makes it impossible to design Hamiltonian machines.
We first study \emph{queue
machines}, where insertion is only possible at the left endpoint. We remind that
it is indeed possible to visit all words of $\{0,1\}^\ell$ 
via left-push and
right-pop using the construction of universal words, but show that this
cannot be achieved by a queue machine (i.e., with constant auxiliary memory), 
even without the constant-delay requirement.
We next study
\emph{stack machines}, where we push and pop at one end only, and show that in a sense such machines also
cannot satisfy the constant-memory requirement.

\subparagraph{Paper outline.} 
In \cref{sec:tape}, we define our model of tape machines and construct one that proves \cref{thm:tape-intro}.
We do the same in \cref{sec:deque} for deque machines and \cref{thm:deque-intro}.
Finally, in \cref{sec:stackqueue} we study weaker computation models (queue machines and stack machines) and show that Hamiltonian machines of this kind do not exist.
We also provide implementations to run our tape and deque machines (See \cite{OurImplementation} and  \cref{app:implementation}).

\section{Tape Machines}
\label{sec:tape}
This section presents \emph{tape machines}, which store the current word on a tape.
We first define tape machines formally. Then, we explain how to construct
such machines that enumerate all binary words of length~$\ell$ in constant delay.
Last, we build a machine which is additionally
\emph{Hamming-1}, i.e., it only flips one bit between two consecutive words, proving \cref{thm:tape-intro}. 

\subparagraph{Machine model.}
A tape machine is intuitively a variant of linear bounded Turing machines~\cite{kuroda1964classes}
 where we additionally distinguish some states as \emph{output states}. Formally,
let $\bow$ and~$\eow$ be fresh symbols to be used as left and right markers.
A \emph{(deterministic) tape machine} $\calM = (Q, q_\ii, q_\ff, Q_\OO, \delta)$
consists of a set $Q$ of \emph{states},
an \emph{initial state} $q_\ii \in Q$,
a \emph{halting state} $q_\ff \in Q$,
a set $Q_\OO \subseteq Q$ of \emph{output states},
and a \emph{transition function} $\delta
\colon (Q \setminus \{q_\ff\}) \times \{0,1, \bow,\eow\} \to Q \times \{0, 1, \bow,
\eow\} \times \{-1, 0, +1\}$. The arguments to the function describe the current state and
the tape symbol read; and the return value describes the next state, the
tape symbol
written, and the motion of the head. We require that when $\calM$  reads
the beginning-of-word marker symbol~$\bow$ then it writes back~$\bow$ and moves to the
right, and likewise for~$\eow$ moving to the left; further, $\calM$ 
never writes the marker symbols otherwise. 
A \emph{configuration} of $\calM$ is a triple consisting of a state, the
contents of the tape, and the position of the head: for brevity we will write them as a
pair indicating the state and the tape contents with the head position
written using a red mark. 
We sometimes use a gray mark to indicate the position of the head in the previous configuration.

For a given word length $\ell\geq 1$, the machine $\calM$ starts
its execution in the initial configuration $c_\ell \coloneq (q_\ii, \bow \omark{0}0^{\ell-1} \eow)$.
The run of $\calM$ from $c_\ell$ is then the deterministic and 
potentially infinite sequence of configurations defined as expected: see Appendix~\ref{apx:semantics-tape} for details.
The only difference with Turing machines is the presence of output
states: each time $\calM$ visits an output state, it \emph{produces} in unit
time the word $w \in \{0,1\}^\ell$ obtained from the current tape by
discarding $\bow$ and $\eow$. (Note that the head position has no impact on the produced word~$w$.) 

\begin{toappendix}
\label{apx:fig-ham1}
  \subsection{Formal Semantics of Tape Machines}
  \label{apx:semantics-tape}
In this section, we formally define the semantics of tape machines.
Let $\calM = (Q, q_\ii, q_\ff, Q_\OO, \delta)$ be a tape machine.
For $w \in \bool^+$ a nonempty word, 
we write $\ell \coloneq |w|$ the length
of~$w$, and we write the bits of $w$ as $w[1], \ldots, w[\ell]$. 
A \emph{configuration}
of~$\calM$ on~$w$ is a triple $c = (q,p,w)$ with $q \in Q$, with $0 \leq p \leq
  |w|+1$, and with $w \in \bool^+$.
We now formally define a computation step of~$\calM$ on configuration $c = (q,p,w)$ by
formally defining the \emph{next configuration} $c'$ of $\calM$.
Let $b \coloneq w[p]$ if $1 \leq p \leq \ell$, and otherwise let $b \coloneq \bow$ if
$p=0$ and $b\coloneq \eow$ if $p=|w|+1$.
Let $(q',b',t)
\coloneq \delta(q,b)$. The \emph{next configuration} is $c' = (q', p+t, w')$,
where $w'$ is obtained from~$w$ by replacing the $p$-th symbol $b = w[p]$
by~$b'$ if $1 \leq p \leq \ell$ and setting $w' \coloneq w$ otherwise. Recall that
from our definition of tape machines, if $p = 0$ then necessarily $t = 1$ and
  $p+t = 1$ (and the ignored value $b'$ must necessarily be $\bow$), and likewise if $p = |w|+1$ then necessarily $t = -1$ and $p+t =
  |w|$ (and $b'$ must necessarily be $\eow$).

We will now formally define the \emph{run} of~$\calM$ for a length $\ell$. 
The \emph{initial configuration} of $\calM$ on~$w$ is $c_\ell \coloneq
(q_\ii,p_0,w_0)$ where $w_0 \coloneq 0^\ell$ is the initial word and where $p_0 \coloneq 1$ is the initial position of the head.
We set the first configuration $c_0$ of the run to be $c_0 \coloneq c_\ell$.
Now, at any step $t \in \NN$ where we reach configuration $c_t = (q_t, p_t, w_t)$, if $q_t$ is an output state ($q_t \in Q_\OO$) then the machine produces (in unit time) the current word~$w_t$. Further, if $q_t$ is the halting state ($q_t
= q_\ff$) then the machine halts, otherwise we define 
$c_{t+1}$ as the next configuration of~$\calM$ on~$c_t$.
This defines a sequence
$c_0, c_1, \ldots$ of configurations, which is infinite if the machine does not
halt.
\end{toappendix}

We study tape machines satisfying some requirements. 
First, a tape machine~$\calM$ is \emph{Hamiltonian} if, for
every length $\ell \geq 1$, 
the run of $\calM$ from $c_\ell$ terminates and produces
each word of $\bool^\ell$
exactly once. 
Note that we can always assume that $\ell > \alpha$ for arbitrary constants $\alpha$, because codes for length $\leq \alpha$ can be hardcoded as a subroutine of~$\calM$. 
Second, a tape machine $\calM$ is \emph{constant-delay} if there is a constant $B \in \NN$ (depending only on~$\calM$ but not on the word length~$\ell$)
that bounds the time before the first output, the delay between two consecutive outputs, and the time between the last output
and the end of the run.

\begin{figure}
    \subfloat[Complete binary tree for words of length $\ell=4$. Its depth
    is~$3=(\ell-1)$. Blue numbers indicate the enumeration order
    ($0000$ is produced first).]{\label{fig:tape-easy-tree}\begin{tikzpicture}%
    \let\hdist\relax%
    \let\vdist\relax%
    \let\vsep\relax%
    \let\yIII\relax%
    \let\yII\relax%
    \let\yI\relax%
    \newcommand{\upsign}{\triangle}
    \newcommand{\downsign}{\triangledown}
    \newlength{\hdist}\setlength{\hdist}{18mm}%
    \newlength{\vdist}\setlength{\vdist}{8mm}%
    \newlength{\vsep}\setlength{\vsep}{0.5mm}%
    \newlength{\yIII}\setlength{\yIII}{0mm}%
    \newlength{\yII}\setlength{\yII}{\dimexpr\vdist/2\relax}%
    \newlength{\yI}\setlength{\yI}{\dimexpr\vdist*3/2\relax}%
    \newcommand{\localrule}{\rule[-.2em]{0pt}{1.1em}}
    
    \path(0,\vdist*3.5) node[tape] (root) {\localrule\omark{1}000};
    \path (root.south) ++(0pt,-1pt) node[order, anchor=north] {\nth{2}};

    \foreach \i in {0,1} {
        \foreach \ii in {0,1} {
            \foreach \iii in {0,1} {
                \path (3*\hdist,\yIII) node[tape] (\i\ii\iii) {\localrule{\color{cborange}\i\ii\iii}\omark{1}};
                \path (\i\ii\iii.east)
                ++(1pt,0pt) node[order,anchor=west] {\nth{\the\numexpr\i*7+\ii*3+\iii+4\relax}}
                ;
                \global\advance \yIII by \dimexpr\vdist\relax
            }
            \path (2*\hdist,\yII) node[tape] (\i\ii) {\localrule{\color{cborange}\i\ii}\omark{1}0};
            \path (\i\ii.south) ++(0pt,-1pt) node[order, anchor=north] {\nth{\the\numexpr\i*7+\ii*3+3\relax}}
            ;
            \path[edge] (\i\ii) to node[near start,sloped,below] {0} (\i\ii0);
            \path[edge] (\i\ii) to node[near start,sloped,above] {1}(\i\ii1);
            \global\advance \yII by \dimexpr\vdist*2\relax
        }
        \path (1*\hdist,\yI) node[tape] (\i) {\localrule{\color{cborange}\i}\omark{1}00};
        \path (\i.north)
            ++(0pt,1pt) node[order,anchor=south] {\nth{\the\numexpr\i*7+9\relax}}
            ;
        \path[edge] (\i) to node[near start,sloped,below] {0} (\i0);
        \path[edge] (\i) to node[near start,sloped,above] {1} (\i1);
        \global\advance \yI by \dimexpr\vdist*4\relax
    }
    \path[edge] (root) to node[near start,sloped,below] {0} (0);
    \path[edge] (root) to node[near start,sloped,above] {1} (1);

\end{tikzpicture}%
}
    \hfill
\subfloat[Definition of the tape machine~$\calT_0$. 
  Red marks indicate the position of the head,
    gray marks the previous position of the head, and $p$ stands for the parity of the depth of the node.]{\label{fig:tape-easy-table}
    \begin{minipage}[c]{.44\linewidth}
        \noindent\input{table/tape-easy}
    \end{minipage}}
    \caption{Illustrations for \cref{prp:t0-correct}}
    \label{fig:tape-easy}
\end{figure}

\subparagraph{Constructing Hamiltonian constant-delay machines.}
It is easy to see that Hamiltonian tape machines exist if we forgo the constant-delay requirement.
For instance, we can follow the usual construction of the BRGC: start from the word $0^\ell$, flip the rightmost bit on odd steps, and flip the bit that is to the left of the rightmost~$1$ on even steps.
This can be implemented by a tape machine whose state stores the parity of the step number. In odd steps, the machine moves the head to the rightmost symbol and flips it. In even steps, the machine moves the head to the rightmost symbol, scans leftwards until it finds the symbol to flip, and flips it.
However, this tape machine has a delay that is not constant, i.e., it is not
independent from the word length~$\ell$.
Our challenge in this section is to build a tape machine that obeys the constant-delay requirement: this is not obvious because it implies in particular that 
tape modifications must be entirely guided by the symbols
at constant distance from the head and based on the current state (which has constant size).
We next present a tape machine that meets these requirements.

\begin{proposition}
  \label{prp:t0-correct}%
  The tape machine $\calT_0$ of \cref{fig:tape-easy-table} is constant-delay and Hamiltonian.
\end{proposition}

This is a first step towards proving \cref{thm:tape-intro}, which we present for pedagogical reasons: it will help understand the Hamming-1 construction presented later, and also introduces key ideas used throughout the paper.
For the sake of clarity, here and in other proofs, we describe the operation of a tape machine by
specifying transitions that read and write multiple symbols in one go: for
some constant $\rho \in \NN$ (independent from the word length), the machine reads and writes the symbols left
and right of the head in a window of size~$\rho$ centered on the head. This can be encoded in the
formal model presented above up to increasing the state space and the delay by
some function of~$\rho$: we can simply hardcode a constant-length sequence of
transitions that reads all $\rho$ symbols around the head and then overwrites them.
We will use $\rho=3$ in the proof below.

\begin{proof}[Proof of \protect\cref{prp:t0-correct}]

Let $\ell\geq 1$ be the word length. 
When the run begins, we produce the word $0^\ell$, and toggle the first
  bit to change the tape contents to $\bow \omark{1}0^{\ell-1} \eow$. We
  then explain how we enumerate the remaining words, called
  \emph{non-zero words} in what follows.

Let us construct a bijection between the non-zero words and the nodes of the complete ordered binary tree $B_{\ell-1}$ of depth
$\ell-1$; we will then explain our enumeration
in terms of the tree.
By distinguishing the rightmost~$1$, each non-zero word can be written as~$w10^k$ for some binary word~$w$ and some number~$k$ of $0$'s.
  Hence the non-zero words are in bijection with the nodes of $B_{\ell-1}$,
  that we illustrate in \cref{fig:tape-easy-tree} for $\ell=4$ (ignoring for now the blue numbers and red marks).
  The root of $B_{\ell-1}$ corresponds to the word $10^{\ell-1}$. Now, for $k\geq 1$, if
$w10^{k}$ is the word corresponding to an internal node $n$, then the 0-child of~$n$ corresponds to the word $w010^{k-1}$ and its 1-child corresponds to the word $w110^{k-1}$. 
  In summary, the word~$w 1
0^k$, where~$0 \leq k\leq \ell-1$ and $|w|+1+k=\ell$, encodes the node in 
  $B_{\ell-1}$
that is reachable from the root by the path~$w$ (reading $w$ from left to
right; $w$ is shown in orange in \cref{fig:tape-easy-tree}).

  We enumerate the nodes of $B_{\ell-1}$ in a way that ensures
  that consecutive nodes are at distance at most $3$ from each
other. We do this via 
the “even-odd trick”~\cite{karaganis1968cube,sekanina1960ordering} (see also~\cite{uno2003two}):
  we do a depth-first traversal of $B_{\ell-1}$
and produce the nodes in prefix or postfix order, depending on the parity of
  the depth (the root being at even depth). Formally, starting
with the root, we do the following on every node $n$ with children $n_0$
  and $n_1$:
  \begin{itemize}
    \item if the depth of $n$
is even then we produce $n$, then recurse into $n_0$, then recurse into~$n_1$;
      \item if the depth is odd then we first recurse into $n_0$, then
        recurse into~$n_1$, and then produce~$n$.
  \end{itemize}
  The order in which nodes are enumerated by this
procedure for $B_{4-1}$
is shown in blue in \cref{fig:tape-easy-tree} 
(starting at index $2$ because we produced the word $0^4$ at the beginning).

The correctness of the machine hinges on the following claim, which is easily proved by induction using the transition table of $\calT_0$ given in \cref{fig:tape-easy}.

\begin{claim}
\label{clm:t0}
    For every word~$w$,~$|w|<\ell$, $\calT_0$ visits exactly two configurations of
    the form $(q,\bow w\omark{1}0^k\eow$): the first time~$q=(\downarrow,p)$, and
    the second time~$q=(\uparrow,p)$, where $p = |w| \bmod 2$.
\end{claim}

 We now give a high level explanation of how~$\calT_0$ performs the traversal on $B_{\ell-1}$.
It maintains its head on the
rightmost $1$, and memorizes with its state the parity of (the depth of the node
corresponding to) the current word, and whether it is currently going up or
down.
Then, one may check that the following elementary tests/operations can be achieved by local reads/writes
in a window of size $\rho=3$ centered on the head.
\begin{enumerate}
  \item Check whether the current
tape encodes a leaf or not, by looking at the right symbol (which is always either $\eow$ or $0$).
\item Check whether the current
tape encodes the root (the symbol $b$ left of the head is $\bow$), a $0$-child ($b$ is~$0$) or a $1$-child ($b$ is~$1$). 
\item
Go from one node to its $0$-child: $w\omark{1}0^k\leadsto w\dimmedomark{0}\omark{1}0^{k-1}$, i.e., set to 0 the bit under
the head, move the head right (the old head is shown in gray), and set to 1 the bit under the head.
\item Go from a $0$-child to its sibling $1$-child:
$w0\omark{1}0^{k}\leadsto w1\omark{1}0^k$,
i.e., flip the bit left of the head.
\item Go from a $1$-child to its parent:  $w1\omark{1}0^k\leadsto w\omark{1}\dimmedomark{0} 0^k$. 
\end{enumerate}
Formally, the tape machine~$\calT_0$ described in 
\cref{fig:tape-easy-table}
starts by outputting~$0^{\ell}$ and reaching the configuration
$\bow \omark{1} 0^{\ell-1}\eow$ via the first transition in the table:
its tape corresponds to the root of~$B_{\ell-1}$ and the state is $(\downstate,\Even)$.
Then,~$\calT_0$ implements the above elementary tests/operations using the four middle transitions, in order to perform a full
traversal of~$B_{\ell-1}$.
This traversal ends when the configuration~$\bow \omark{1} 0^{\ell-1}\eow$ is reached again, this time with state~$(\upstate,\Even)$, and then we take the last transition and end the run.
\end{proof}

By analyzing the sequence of transitions between two words produced by $\calT_0$, one can observe that the maximal distance between flipped bits is two. Further, the
code changes at most~3 bits from one word produced to the next (between the
\nth{9} word and the \nth{10} for instance). 
Because $3 > 1$, this code is
technically not a Gray code but what some authors call a \emph{quasi-Gray
code}~\cite{bose2010improved}. 
We now show how to fix this and build a \emph{Hamming-1} machine.

\begin{toappendix}
\begin{figure}
    \begin{tikzpicture}
    \tikzset{
      edge/.style={-latex, black, draw, thick},
      tape/.style={rounded rectangle, draw=black,inner sep=2pt,minimum size=4pt,outer sep=2pt},
      order/.style={blue,font=\footnotesize, inner sep=0pt, outer sep=0pt},
    }

\node[tape] (01000) at (0, 7.75) {{\color{orange!75!black}}0\omark{1}000};
\node[order, below of = 01000, node distance = 0.45 cm] {\nth{2}} ;
\node[tape] (01100) at (4.0, 3.75) {{\color{orange!75!black}0}1\omark{1}00};
\node[order, below of = 01100, node distance = 0.45 cm] {\nth{3}} ;
\path[edge, shorten >= 0.8 cm] (01000.east) to node[near start, sloped, above] {0} ($(01100.north) + (0,-0.1)$);
\node[tape] (01110) at (8.0, 1.75) {{\color{orange!75!black}01}1\omark{1}0};
\node[order, below of = 01110, node distance = 0.45 cm] {\nth{4}} ;
\path[edge, shorten >= 0.8 cm] (01100.east) to node[near start, sloped, above] {1} ($(01110.north) + (0,-0.1)$);
\node[tape] (01111) at (12.0, 0.75) {{\color{orange!75!black}011}1\omark{1}};
\node[order, below of = 01111, node distance = 0.45 cm] {\nth{5}} ;
\path[edge, shorten >= 0.8 cm] (01110.east) to node[near start, sloped, above] {1} ($(01111.north) + (0,-0.1)$);
\node[tape] (01101) at (12.0, 1.25) {{\color{orange!75!black}011}0\omark{1}};
\node[order, above of = 01101, node distance = 0.45 cm] {\nth{6}} ;
\node[tape] (01001) at (12.0, 2.75) {{\color{orange!75!black}010}0\omark{1}};
\node[order, below of = 01001, node distance = 0.45 cm] {\nth{7}} ;
\node[tape] (01011) at (12.0, 3.25) {{\color{orange!75!black}010}1\omark{1}};
\node[order, above of = 01011, node distance = 0.45 cm] {\nth{8}} ;
\node[tape] (01010) at (8.0, 2.25) {{\color{orange!75!black}01}0\omark{1}0};
\node[order, above of = 01010, node distance = 0.45 cm] {\nth{9}} ;
\path[edge, shorten > = 0.8cm] (01010.east) to node[near start, sloped, above] {0} ($(01011.south) + (0,0.1)$);
\node[tape] (00010) at (8.0, 5.75) {{\color{orange!75!black}00}0\omark{1}0};
\node[order, below of = 00010, node distance = 0.45 cm] {\nth{10}} ;
\node[tape] (00011) at (12.0, 4.75) {{\color{orange!75!black}000}1\omark{1}};
\node[order, below of = 00011, node distance = 0.45 cm] {\nth{11}} ;
\path[edge, shorten >= 0.8 cm] (00010.east) to node[near start, sloped, above] {0} ($(00011.north) + (0,-0.1)$);
\node[tape] (00001) at (12.0, 5.25) {{\color{orange!75!black}000}0\omark{1}};
\node[order, above of = 00001, node distance = 0.45 cm] {\nth{12}} ;
\node[tape] (00101) at (12.0, 6.75) {{\color{orange!75!black}001}0\omark{1}};
\node[order, below of = 00101, node distance = 0.45 cm] {\nth{13}} ;
\node[tape] (00111) at (12.0, 7.25) {{\color{orange!75!black}001}1\omark{1}};
\node[order, above of = 00111, node distance = 0.45 cm] {\nth{14}} ;
\node[tape] (00110) at (8.0, 6.25) {{\color{orange!75!black}00}1\omark{1}0};
\node[order, above of = 00110, node distance = 0.45 cm] {\nth{15}} ;
\path[edge, shorten > = 0.8cm] (00110.east) to node[near start, sloped, above] {1} ($(00111.south) + (0,0.1)$);
\node[tape] (00100) at (4.0, 4.25) {{\color{orange!75!black}0}0\omark{1}00};
\node[order, above of = 00100, node distance = 0.45 cm] {\nth{16}} ;
\path[edge, shorten > = 0.8cm] (00100.east) to node[near start, sloped, above] {0} ($(00110.south) + (0,0.1)$);
\node[tape] (10100) at (4.0, 11.75) {{\color{orange!75!black}1}0\omark{1}00};
\node[order, below of = 10100, node distance = 0.45 cm] {\nth{17}} ;
\node[tape] (10110) at (8.0, 9.75) {{\color{orange!75!black}10}1\omark{1}0};
\node[order, below of = 10110, node distance = 0.45 cm] {\nth{18}} ;
\path[edge, shorten >= 0.8 cm] (10100.east) to node[near start, sloped, above] {0} ($(10110.north) + (0,-0.1)$);
\node[tape] (10111) at (12.0, 8.75) {{\color{orange!75!black}101}1\omark{1}};
\node[order, below of = 10111, node distance = 0.45 cm] {\nth{19}} ;
\path[edge, shorten >= 0.8 cm] (10110.east) to node[near start, sloped, above] {1} ($(10111.north) + (0,-0.1)$);
\node[tape] (10101) at (12.0, 9.25) {{\color{orange!75!black}101}0\omark{1}};
\node[order, above of = 10101, node distance = 0.45 cm] {\nth{20}} ;
\node[tape] (10001) at (12.0, 10.75) {{\color{orange!75!black}100}0\omark{1}};
\node[order, below of = 10001, node distance = 0.45 cm] {\nth{21}} ;
\node[tape] (10011) at (12.0, 11.25) {{\color{orange!75!black}100}1\omark{1}};
\node[order, above of = 10011, node distance = 0.45 cm] {\nth{22}} ;
\node[tape] (10010) at (8.0, 10.25) {{\color{orange!75!black}10}0\omark{1}0};
\node[order, above of = 10010, node distance = 0.45 cm] {\nth{23}} ;
\path[edge, shorten > = 0.8cm] (10010.east) to node[near start, sloped, above] {0} ($(10011.south) + (0,0.1)$);
\node[tape] (11010) at (8.0, 13.75) {{\color{orange!75!black}11}0\omark{1}0};
\node[order, below of = 11010, node distance = 0.45 cm] {\nth{24}} ;
\node[tape] (11011) at (12.0, 12.75) {{\color{orange!75!black}110}1\omark{1}};
\node[order, below of = 11011, node distance = 0.45 cm] {\nth{25}} ;
\path[edge, shorten >= 0.8 cm] (11010.east) to node[near start, sloped, above] {0} ($(11011.north) + (0,-0.1)$);
\node[tape] (11001) at (12.0, 13.25) {{\color{orange!75!black}110}0\omark{1}};
\node[order, above of = 11001, node distance = 0.45 cm] {\nth{26}} ;
\node[tape] (11101) at (12.0, 14.75) {{\color{orange!75!black}111}0\omark{1}};
\node[order, below of = 11101, node distance = 0.45 cm] {\nth{27}} ;
\node[tape] (11111) at (12.0, 15.25) {{\color{orange!75!black}111}1\omark{1}};
\node[order, above of = 11111, node distance = 0.45 cm] {\nth{28}} ;
\node[tape] (11110) at (8.0, 14.25) {{\color{orange!75!black}11}1\omark{1}0};
\node[order, above of = 11110, node distance = 0.45 cm] {\nth{29}} ;
\path[edge, shorten > = 0.8cm] (11110.east) to node[near start, sloped, above] {1} ($(11111.south) + (0,0.1)$);
\node[tape] (11100) at (4.0, 12.25) {{\color{orange!75!black}1}1\omark{1}00};
\node[order, above of = 11100, node distance = 0.45 cm] {\nth{30}} ;
\path[edge, shorten > = 0.8cm] (11100.east) to node[near start, sloped, above] {1} ($(11110.south) + (0,0.1)$);
\node[tape] (11000) at (0.0, 8.25) {{\color{orange!75!black}}1\omark{1}000};
\node[order, above of = 11000, node distance = 0.45 cm] {\nth{31}} ;
\path[edge, shorten > = 0.8cm] (11000.east) to node[near start, sloped, above] {1} ($(11100.south) + (0,0.1)$);
\end{tikzpicture}
    \caption{Example run of $\calT_1$ on words of length 5, where the bottom word and top word of each node is written at the top and bottom of the node. Pay attention to the fact that the 0-children and 1-children of nodes are not always ordered in the same way, to match the order in which $\calT_1$ enumerates the words}
    \label{fig:tape-ham1-tree}
\end{figure}
\end{toappendix}

\begin{figure}[t]
    \newlength{\leftfigsize}\setlength{\leftfigsize}{.39\linewidth}
    \newlength{\rightfigsize}\setlength{\rightfigsize}{\dimexpr\linewidth-\leftfigsize-\columnsep\relax}
  \subfloat[Definition of the tape machine~$\calT_1$]{\label{fig:tape-ham1-table}%
        \input{table/tape-ham1}%
    }%
    \hfill%
    \subfloat[Visual help for inductions]{\label{fig:tape-ham1-induction}%
    \adjustbox{width=\rightfigsize}{\begin{tikzpicture}%
    \input{figure/style}%
    \let\hdist\relax\let\vdist\relax%
    \newlength{\hdist}\setlength{\hdist}{30mm}%
    \newlength{\vdist}\setlength{\vdist}{30mm}%
    \newcommand{\localrule}{\rule[-.3em]{0pt}{1em}}%
    \coordinate (root) 
        node[tape,anchor=north] (root-south) {\localrule${\color{darkorange}u}x\omark{1}{\color{black}00^{k}}$}
        node[tape,anchor=south] (root-north) {\localrule${\color{darkorange}u}\overline{x}\omark{1}{\color{black}\dimmedomark{0}0^{k}}$};

    \path[edge,dashed] (root-south.west) ++(-8mm,0) to (root-south.west);
    \path[edge,dashed,shorten >=0] (root-north.west) to ++(-8mm,0);
    
    \path (root) ++ (\hdist,-.5\vdist) coordinate (child1) 
        node[tape,anchor=north] (child1-south) {\localrule${\color{darkorange}u x} \dimmedomark{1} \omark{1} {\color{black}0^k}$}
        node[tape,anchor=south] (child1-north) {\localrule${\color{darkorange} u x}0 \omark{1}{\color{black}0^k}$};
        
    \path (root) ++ (\hdist,.5\vdist) coordinate (child2) 
        node[tape,anchor=north] (child2-south) {\localrule${\color{darkorange}u \overline{x}}0\omark{1}{\color{black}0^k}$}
        node[tape,anchor=south] (child2-north) {\localrule${\color{darkorange}u\overline{x}}1\omark{1}{\color{black}0^k}$};

    \path[draw] ($0.5*(child1-south.east)+0.5*(child1-north.east)$) -- ++(.75*\hdist,10mm) coordinate (tree1-ne) -- ++(0,-20mm) coordinate (tree1-se) --cycle;
    \path[draw] ($0.5*(child2-south.east)+0.5*(child2-north.east)$) -- ++(.75*\hdist,10mm) coordinate (tree2-ne) -- ++(0,-20mm) coordinate (tree2-se) --cycle;
    
    \path[edge] (root-south) to[out=0, in=180,looseness=1.7] (child1-south);
    \path[edge,dashed,rounded corners=10pt] (child1-south.east)
        to 
        ($(tree1-se)+(6.6pt,-10pt)$)
        to
            node[midway, rotate=90,anchor=north,outer sep=5pt,inner sep=0pt] {\begin{tabular}{c}Use induction\\[-.3ex]hypothesis\end{tabular}} 
        ($(tree1-ne)+(6.6pt,10pt)$)
        to (child1-north.east);
    \path[edge] (child1-north.north west) to[bend left] (child2-south.south west);
    \path[edge,dashed,rounded corners=10pt] (child2-south.east) 
        to ($(tree2-se)+(6.6pt,-10pt)$)
        to
            node[midway, rotate=90,anchor=north,outer sep=5pt,inner sep=0pt] {\begin{tabular}{c}Use induction\\[-.3ex]hypothesis\end{tabular}} 
        ($(tree2-ne)+(6.6pt,10pt)$)
        to (child2-north.east);
    \path[edge] (child2-north) to[out=180, in=0,looseness=1.7] (root-north);
\end{tikzpicture}%
}%
    }
    \caption{Illustrations for \cref{prp:t1-correct}}
    \label{fig:tape}
\end{figure}

\subparagraph{Hamming-1 machines.}
A tape machine $\calM$ is \emph{Hamming-1} if,
for any $\ell \geq 1$, 
for any two consecutive words $w$ and $w'$ that are
produced
in the run of~$\calM$ on $0^\ell$,
then the Hamming distance between $w$ and $w'$ is~$1$ (i.e., there is a single bit that changes).

Whenever a Gray code is implemented by a constant-delay tape machine,
then the maximal distance between flipped bits is bounded by
some constant $k$, a property called \emph{$k$-skew-tolerance} \cite{sachimelfarb2025improved}.
Most Gray codes, such as the BRGC, are not skew-tolerant, 
and the existence of 1-skew-tolerant Gray codes is an open question %
\cite[P10]{Mutze23-gray}.
However, \cite{sachimelfarb2025improved} give examples of $k$-skew-tolerant Gray codes for~$k=3$ and~$k=2$; the code for~$k=3$, called \emph{construction A}, is of particular interest to us.
Indeed, we now show how to construct a machine that is Hamming-1, Hamiltonian and constant-delay, which turns out to implement a Gray code that is closely related to construction A; see Appendix~\ref{app:hs24} for details.

\begin{proposition}
  \label{prp:t1-correct}
  The tape machine~$\calT_1$, defined in \cref{fig:tape-ham1-table},
  is constant-delay, Hamiltonian and Hamming-1.
\end{proposition}

\begin{proof}

  Let $\ell \geq 2$ be the word length.
  We ignore for now the two words $0^\ell$ and $10^{\ell-1}$ that are enumerated first and last, respectively. All other words are of the form
  $\hlpath{u}\hlpari{x}10^{k}$ for $x\in \{0,1\}$ and $u\in \{0,1\}^{\ell-k-2}$, by distinguishing the rightmost
  $1$ and the bit $x$ that precedes it.
  As it was the case for~$\calT_0$, the machine~$\calT_1$ maintains its head on the rightmost~$1$, as can be seen by inspecting its transition table (\cref{fig:tape-ham1-table}).
  We use the notation $\bar{0} = 1$ and $\bar{1} = 0$.   %

  Consider the complete ordered binary tree $B_{\ell-2}$ of depth $\ell-2$.
  To every node
  $n$ of depth~$d$, we associate two words of length $\ell$ : the \emph{bottom word} $w_n^\bot$ and the \emph{top
  word} $w_n^\top$. 
  They have the form $w_n^\bot = \hlpath{u_n} \hlpari{x_n} \omark{1}0^{k}$ and $w_n^\top = \hlpath{u_n}\hlpari{\overline{x_n}}\omark{1}0^{k}$ where~$d+2+k=\ell$, and where $x_n\in \{0,1\},u_n \in \{0,1\}^d$ are defined as follows:
  $\hlpath{u_n}$ is the path from the root to~$n$ and~$\hlpari{x_n}$ is the parity of the number of~$0$'s in~$u_n$.
  For instance, the bottom word associated with the root $r$ is $w_r^\bot =
  \hlpari{0}\omark{1}0^{\ell-2}$ (i.e., $u_r$ is the empty word and~$x_r=0$);
  and among~$\{\hlpath{011}\hlpari{0}\omark{1},\hlpath{011}\hlpari{1}\omark{1}\}$, that are associated with the same node~$n$, $\hlpath{011}\hlpari{1}\omark{1}$ is the bottom word (i.e., $x_n=1$ and $u_n=011$).
   We show $B_{\ell-2}$ in Figure~\ref{fig:tape-ham1-tree} in Appendix~\ref{apx:fig-ham1}, with each tree node indicating the top word and bottom word one above the other.
  
  Let~$n$ be any internal node having bottom word $w_n^\bot = \hlpath{u_n} \hlpari{x_n} \omark{1} 0^{k+1}$ and top word $w_n^\top = \hlpath{u_n} \hlpari{\overline{x_n}} \omark{1} 0^{k+1}$.
  We let~$c$ denote the $x_n$-child of~$n$ and~$\bar{c}$ its $\overline{x_n}$-child.
  This implies that we have $u_c=u_n x_n$, we have $u_{\bar{c}}=u_n \overline{x_n}$, and we always have $x_c = 1$ and~$x_{\bar c} = 0$.

  Then, it may be observed that the following word pairs are at Hamming distance~$1$:
  \begin{enumerate}
      \item $w_c^\bot$ and $w_n^\bot$, that is the \textbf{bottom} word of~$c$ and the \textbf{bottom} word of its parent;
      \item $w_c^\top$ and $w_{\bar{c}}^{\bot}$; that is the \textbf{top} word of~$c$ and the \textbf{bottom} word of its sibling;
      \item $w_{\bar{c}}^\top$ and $w_n^\top$, that is the \textbf{top} word of~$\bar{c}$ and the \textbf{top} word of its parent;
      \item $w_{\lambda}^\bot$ and $w_{\lambda}^\top$,  for every \textbf{leaf} $\lambda$ (remark that this is true for every node more generally).
  \end{enumerate}
  These claims are the basis of a traversal of $B_{\ell-2}$, inductively defined as follows.
  Starting from the root, and letting $n$ be the
  current node, we first produce the word $w_n^\bot$, then recurse into the $x_n$-child of $n$, then recurse into the $\overline{x_n}$-child of $n$, then produce the word
  $w_n^\top$. 
  Using a simple induction (item 4 for the base case and 1-3 for the induction;  see also \cref{fig:tape-ham1-induction}),
  one can show that this traversal induces a Gray code, when the missing words $0^\ell$ and $10^{\ell-1}$ are produced first and last, as we do. 

  The tape machine $\calT_1$ then produces exactly this code, as follows from:

  \begin{claim*}
      For every node $n$ of $B_{\ell-2}$ with bottom word $\hlpath{u_n}\hlpari{x_n}\omark{1}0^k$, if we start
  the machine on configuration $(\downarrow,~\bow \hlpath{u_n}\hlpari{x_n} \omark{1} 0^k \eow)$ then it will produce the words
  of the above traversal on the subtree rooted at $n$, and it will reach configuration 
  $(\uparrow,~\bow \hlpath{u_n} \hlpari{\overline{x_n}} \omark{1} 0^k \eow)$.
  \end{claim*}
  It is routine to prove
  this claim by induction on $B_{\ell-2}$ (see
  \cref{fig:tape-ham1-induction}), inspecting the transitions of $\calT_1$.
  All that remains is to check the behavior on
  the root $r$ of $B_{\ell-2}$:
  \begin{itemize}
      \item When $\calT_1$ starts, it first produces the word $0^\ell$ and then it reaches the configuration $(\downstate,~\bow \hlpari{0}\omark{1}0^{\ell-2} \eow)$, that is with tape~$w_{r}^\bot$ with its head on the rightmost~$1$. %
      \item From the configuration $(\upstate,~\bow \hlpari{1}\omark{1}0^{\ell-2} \eow)$,
      that is with tape~$w_r^\top$ with its head on the rightmost~$1$, the machine $\calT_1$ produces $10^{\ell-1}$ and then halts.\qedhere
  \end{itemize}
\end{proof}

\begin{toappendix}

\subsection{Link with Construction A of Sac Himelfarb and Schwartz \texorpdfstring{\cite{sachimelfarb2025improved}}{(2025)}}
\label{app:hs24}

As explained in the main text, the code implemented by tape machine~$\calT_1$
(\cref{fig:tape-ham1-table}) is related with the skew-tolerant code with~$k=3$ given in \cite{sachimelfarb2025improved}.
To that end, let us give inductive definitions of that code ($A_\ell$, below) 
and an inductive definition of the code implemented by~$\calT_1$ ($B_\ell$, below). Let us start with a few definitions.
For a word $w=w[1]\cdots w[m] \in \{0,1\}^m$ let us write
$\mirror(w)$ the word $w[m] \cdots w[1]$.
Let $X = \begin{bmatrix} w_1 & \cdots & w_p\end{bmatrix}^\intercal$ be a column vector of words.
We write $\mirror(X)=\begin{bmatrix} \mirror(w_1) & \cdots &
  \mirror(w_p)\end{bmatrix}^\intercal$ the vector where each word has been mirrored,
and we write $\overleftarrow{X} = \begin{bmatrix} w_p & \cdots &
w_1\end{bmatrix}^\intercal$ the \emph{reversed} vector.
We also let $X[1{:}] = \begin{bmatrix} w_2 & \cdots &
w_p\end{bmatrix}^\intercal$ be the vector obtained by removing the first word from~$X$.
Let~$A_{\ell}$ and~$B_{\ell}$ be the families of vectors inductively defined as follows.
\begin{equation*}
    A_1 = \sparenv*{\begin{array}{cc} 0\\1 \end{array}}
      \quad\text{and}\quad
    A_{\ell+1} =
    \sparenv*{
  \begin{array}{c|c}
    0\cdots 0 & 0 \\
    \hline
     & 1 \\
    A_{\ell} & \svdots \\
     & 1 \\
     \hline
     & 0 \\
    \smash{\overleftarrow{A_{\ell}[1{:}]}} & \svdots \\
     & 0 
    \end{array}
    }
    \qquad\quad
    B_1 = \sparenv*{
      \begin{array}{c}
        0\\
        1
      \end{array}}
      \quad\text{and}\quad
    B_{\ell+1} =
    \sparenv*{
    \begin{array}{c|c}
      0 & 0 \cdots 0 \\
      \hline
      0 & \\
      \svdots & \smash{\overleftarrow{B_\ell[1{:}]}} \\
      0 & \\
      \hline
      1 & \\
      \svdots & \smash{B_\ell[1{:}]} \\
      1 & \\
      \hline
      1 & 0\cdots 0
    \end{array}
    }
\end{equation*}

As can be seen above, the recursive definitions of~$A_\ell$ and~$B_\ell$ are very similar, and one may then show the following by induction.
\begin{lemma}
  \label{lem:SH24-equiv}
  For every~$\ell\geq 1$, we have $B_{\ell} =
    \sparenv*{
    \begin{array}{c}
      0^\ell \\[3pt]
      \mirror\left(\overleftarrow{A_{\ell}[1{:}]}\right)
    \end{array}
    }$.
\end{lemma}

Now, let us show that the code~$B_\ell$ is precisely the one implemented by our tape machine~$\calT_1$.

\begin{proposition}
  \label{prp:tape-inductive}
  For $\ell\geq 2$, the machine~$\calT_1$ produces the code $B_{\ell}$.
\end{proposition}
\begin{proof}
For $X = \begin{bmatrix} w_1 & \cdots & w_p\end{bmatrix}^\intercal$ a column
vector of binary words with $p\geq 3$, we write $X[1:-1] = \begin{bmatrix}
w_2 & \cdots & w_{p-1}\end{bmatrix}^\intercal$ the same vector where we have
removed the first and last words. Moreover, for $u$ a word we write $uX = \begin{bmatrix} u w_1 & \cdots & u w_p\end{bmatrix}$.
We claim the following:
  for any word $u$ and $x\in \{0,1\}$ and $k\geq 0$,  if we start executing
  the machine $\calT_1$ in configuration $(\downarrow,~\bow u x \omark{1} 0^k \eow)$, the machine will reach the configuration $(\uparrow,~\bow u \bar{x}
  \omark{1} 0^k \eow)$ and between these two configurations it will produce either the sequence $u B_{k+1}[1:-1]$
  if~$x=0$, or the sequence $ u \overleftarrow{B_{k+1}[1:-1]}$ if~$x=1$. This
  property can tediously be shown by induction on $k$, referring to
the tree from \cref{fig:tape-ham1-tree,fig:tape-ham1-induction}.
  But then, applying it to
  $k=\ell-2$ gives the desired result: indeed, when the machine starts in the initial state on tape content $\bow
  \omark{0} 0^{\ell-1} \eow$, it will first produce $0^\ell$ (which is the first entry of $B_\ell$) and move to $\bow 0\omark{1}0^{\ell-2} \eow$ on state
  $\downarrow$, then by this property it will generate $B_\ell[1:-1]$ and reach
  $\bow 1 \omark{1} 0^{\ell-2} \eow$, after which it will produce $10^{\ell-1}$ (which is the
  last entry of $B_\ell$) and halt.
\end{proof}

\begin{figure}[h]
    \centering
\begin{tabular}{c@{\hskip\dimexpr\tabcolsep/2\relax}c@{}c@{}ccc@{\hskip\dimexpr\tabcolsep/2\relax}c@{}c@{}cl}
    \cline{1-10}
    \multicolumn{10}{l}{States: $\{\qi,\upstate,\downstate,\qh\}$}
    \\
    \multicolumn{10}{l}{$\text{Initial state~} \qi, \text{~halting state~} \qh$} \\ %
    \multicolumn{10}{l}{Output states: $\{\qi,\upstate,\downstate\}$} \\
    \cline{1-10}
    \multicolumn{10}{l}{Transition table:}
    \\
    $\qi$ &$\thincdots$ &$0\omark{0}\eow\ $& & 
        $\leadsto$ & 
        $\downstate$&$\thincdots$&$0\omark{\color{cborange}1}\eow$\ &
    \\
    $\downstate$ &$\thincdots$ &$0\omark{1}\eow\ $& & 
        $\leadsto$ & 
        $\downstate$&$\thincdots$&$\ \omark{\color{cborange}1}\dimmedomark{1}\eow$&
    \\[5pt]
    $\downstate$ & $\thincdots$&$0\omark{1}xz$&$\thincdots$ & 
        $\leadsto$ & 
        $\downstate$&  $\thincdots$&$\omark{\color{cborange}1}\dimmedomark{1}xz$&$\thincdots$  
    \\ 
      $\downstate$ &&$\bow\omark{1}xz $&$\thincdots$& 
      $\leadsto$ &$\upstate$ && $\bow\omark{1}{\color{cborange}\bar{x}}z$&$\thincdots$  
      \\ 
      $\upstate$ & $\thincdots$&$y\omark{1}0x$&$\thincdots$ & $\leadsto$ & $\downstate$&  $\thincdots$&$y\omark{1}0{\color{cborange}\bar{x}}$&$\thincdots$  
      \\ 
      $\upstate$ & $\thincdots$&$y\omark{1}1x$&$\thincdots$ & $\leadsto$ & $\upstate$ &  $\thincdots$&$ y\dimmedomark{\color{cborange}0}\omark{1}x$&$\thincdots$  
      \\[5pt]
      $\upstate$ &$\thincdots$&$0\omark{1}0\eow$ && $\leadsto$ & $\qh$ &$\thincdots$&$0{\color{cborange}0}\omark{0}\eow$ &
      \\
      \cline{1-10}
      \multicolumn{10}{r@{}}{\footnotesize where $x\in\{0,1\}, y\in\{0,\bow\}, z\in\{0,1,\eow\}$}
    \end{tabular}

    \caption{Definition of the tape machine~$\calT_2$, which enumerates the code~$A$ from \cite{sachimelfarb2025improved}}\label{fig:tape-hs24-table}%
\end{figure}

Unsurprisingly, the code $A$ can also be
implemented in an Hamiltonian constant-delay tape-machine, as stated below.

\begin{proposition}
The tape machine~$\calT_2$, defined in \cref{fig:tape-hs24-table},
  is constant-delay, Hamiltonian, Hamming-1 and implements the code~$A$.
\end{proposition}

\end{toappendix}
\subparagraph{Connections to the DAT model.}
Our results relate to known algorithms
and lower bounds on the efficient implementation of quasi-Gray codes in the
\emph{Decision Assignment Tree} (DAT) model or bit-probe model
\cite{fredman1978observations,bose2010improved,brodal2014integer,rahman2010integer,raskin2017linear}.
In this model, the read and write operations performed on the memory are specified
following a decision tree. However, the DAT model is a black-box model which
does not restrict the computation that the machine performs, or the position in the word of the bits that are read and written. By contrast, in our model of tape machines, the next reads and writes must be decided by following a fixed transition table
(independent from the current word length), and they must happen close to where the head is located. For this reason, the known algorithms for
(quasi-)Gray codes in the DAT model do not imply the existence of tape machines that satisfy our requirements.
Conversely, translating tape machines to the DAT model does not improve on the state of the art in that setting, because our tape machines implicitly store the position of their head: in the DAT model this
amounts to $\lceil \log \ell \rceil$ extra bits that can be arbitrarily read and written at each step. 
These extra bits are why tape machines do not contradict the known
$\Omega(\log \ell)$ lower bound on the number of bits read at each step in the DAT model (shown in Fredman~\cite{fredman1978observations} for Gray codes and claimed by Raskin~\cite{raskin2017linear} for quasi-Gray codes). 

\subparagraph{Decrements and connection to counters in the RAM model.}
Our results also relate to a line of work that has studied the efficient implementation of counters in the RAM model,
where algorithms are run on a machine featuring memory
words of logarithmic size. In this model, \cite{elmasry2013place} presents a counter
that can perform increments and decrements in $O(1)$ time and using $O(1)$
additional memory words; this answered an open problem stated by
Demaine~\cite{demaine}. The tape machines that we presented 
can easily be implemented 
as a RAM algorithm, where we store the current word in a table of $\ell$ bits and keep one auxiliary
word of $\lceil \log \ell \rceil$ bits for the head position and a constant number of extra bits for the state. The increment operations are specified by the machines $\calT_0$ and $\calT_1$ and they can be performed in constant time. Further, $\calT_0$'s and
$\calT_1$'s transition tables can be checked to be injective, so that we can also perform decrements in constant time simply by reversing the transitions. 
Hence, our constructions provide an alternative solution to Demaine's problem, using a single
auxiliary $\lceil \log \ell \rceil$-bit word (instead of two in \cite{elmasry2013place}), and 
additionally ensuring that the successive counter values form a Gray code. Another advantage is that our counter support integers of any size, since we can go from size $l$ bits to $l+1$ in constant time with the same machine. Hence, we could use it to replace an ad-hoc technical construction in~\cite[Section 6]{capelli2023geometric} used to detect large values of increasing size with constant overhead. 
We further discuss in Appendix~\ref{apx:rankunrank} how to implement the rank and unrank operations.

\begin{toappendix}
\subsection{Rank and Unrank Functions}
\label{apx:rankunrank}
In this appendix, 
we discuss the question of connecting the enumeration order on $\{0,1\}^\ell$ defined by our machines to the corresponding integers in $\{0,\ldots, 2^\ell-1\}$. 
We do so with the \textsf{rank} function (compute a given word's position in the
enumeration), and the inverse \textsf{unrank} function (compute the word
produced at a given position).
We focus on the case of $\calT_1$, for which the rank and unrank functions are easily to define, similarly to
\cite[Construction~A]{sachimelfarb2025improved}.
We sketch them below and provide implementations for them in~\cite{OurImplementation} (see also Appendix~\ref{app:implementation}).

We start with \textsf{rank}, where we compute the index of a given word $w \in \{0,1\}^\ell$ in the enumeration.
The special cases~$w=0^{\ell}$
and~$w=10^{\ell-1}$ are treated independently.
Otherwise, we write~$w$ as~$w=u10^k$ for some~$u=a_{1}\cdots a_m\in\{0,1\}^*$ and~$k=\ell-m-1$.
Notice that~$u$ ends with~$a_m$, which is the parity bit of~$w$.
Intuitively, the rank of $w$ is deduced from $u$, which provides the path from the root to the corresponding node in $B_{\ell-2}$. The main technical hurdle comes from the fact that the 
child visited first during the traversal of~$B_{\ell-2}$ may either be the $0$-child or the $1$-child, depending on the parity of the number of $0$'s (see the proof of \cref{prp:t1-correct} and \cref{fig:tape-ham1-tree}). To overcome this, we define $v=b_{1}\cdots b_{m}$, the \emph{corrected} word, which satisfies that~$b_i=0$ when~$a_i$ corresponds to a child taken first during the traversal of~$B_{\ell-2}$. This is formally defined as follows:

For every~$0 \leq i < m$, we denote by~$p_i$ the parity of the number of 0 in~$a_{1}\cdots a_{i}$, i.e., in the prefix of length~$i$ of~$u$.
Then~$v=b_{1}\cdots b_{m}$ is defined as follows:
for every~$0 \leq i< m$, we set $b_i\coloneq a_i$ if $p_{i-1}$ is $\Even$, or $b_i\coloneq \overline{a_i}$ otherwise.
Then, it may be verified that:
\[\textsf{rank}(w) = m+\left(\sum_{i=1}^{m-1} b_i\, (2^{\ell-i}-2)\right)+b_m\,(2^{\ell+1-m}-3)\]

The \textsf{unrank} function computes a word of $\{0,1\}^\ell$ from a given index $n \in {\{0,\ldots,2^\ell-1\}}$, and works by reversing the process described above.
Let $\ell\geq2$ be an input word length. We assume $n\in\{1,\ldots,2^{\ell}-2\}$, with the cases~$n=0$ and $n=2^{\ell}-1$ being treated independently.
Then, observe that, in the paragraph above, the computation of the corrected word from the input word is easy to reverse.
Thus, it suffices to compute the corrected word, which is done by calling $f(n,\ell)$ using the function $f$ defined below (with $\cdot$ denoting word concatenation):
\begin{equation*}
    \forall \ell\geq2, \forall n\in\{1,\ldots,2^{\ell}-2\}\quad f(n,\ell)=
    \begin{cases}
        0 & \text{if }n=1 \\
        0 \cdot f\big(n-1,\ell-1\big) & \text{if }2\leq n\leq 2^{\ell-1}-1 \\
        1 \cdot f\big(n-(2^{\ell-1}-1),\ell-1\big) & \text{if }2^{\ell-1} \leq n\leq 2^{\ell}-3\\
        1 & \text{if }n=2^{\ell}-2
    \end{cases}
\end{equation*}
\end{toappendix}

\section{Deque Machines}
\label{sec:deque}
In the previous section, we have implemented Gray codes with a constant-delay tape machine. These machines only have constantly many states, but they arguably use more than constant memory: their configuration also stores the position of the head, which would require a 
logarithmic number of additional bits.
In this section, we investigate how to produce Gray codes in a different
computation model called \emph{deque machines}, where the additional memory used is really constant.
Intuitively, the current word is stored on a \emph{deque},
i.e., a doubly-linked list allowing us to perform accesses, insertions, and deletions at the two endpoints of the word. In this model we are not concerned with how the doubly-linked list would be implemented on hardware, and we simply assume that these endpoint operations take unit time.
The machine has a constant-sized internal state, but it does not feature a
head: it can only access the symbols that are at the endpoints of the deque.

Like in the previous section, our goal is to construct a deque machine which is Hamiltonian and constant-delay. A necessary condition for the existence of such machines is the existence of so-called \emph{push-pop quasi-Gray codes}, i.e., codes that visit each word of $\{0,1\}^\ell$ exactly once by producing each word from the previous one via push and pop operations, the number of such operations being constant (independent of~$\ell$). Such codes are indeed known to exist, e.g., from universal words~\cite{chung1992universal} or from Hamiltonian paths in the Shuffle Exchange network~\cite{feldmann1996shuffle}: see Appendix~\ref{sec:push-pop} for details. Hence, we will investigate in this section whether push-pop quasi-Gray codes can be implemented by machines with constant auxiliary memory, with the goal of proving \cref{thm:deque-intro} from the introduction.

\begin{toappendix}
\subsection{Push-Pop Quasi-Gray Code}
\label{sec:push-pop}

In this appendix, we give a self-contained argument for the existence
of a \emph{push-pop quasi-Gray code}. This should be contrasted with the deque
machines given in the main text, which give a constructive algorithm to produce
such codes with constant delay and constant additional memory.

A \emph{1-push-pop Gray code} for length $\ell \in \NN$ is a sequence $S_\ell$ in bijection with $\{0,1\}^\ell$ such that each word can be obtained by one push and one pop from the previous word, i.e., for each word $u$ of $S_\ell$ except the first, letting $v$ be the previous word, we can write $v \coloneq w a$ and $u = b w$ for some $a, b \in \{0,1\}$. Remark that a 1-push-pop Gray code is in particular a push-pop quasi-Gray code.

We can show the existence of such codes by a simple argument using
universal words:

\begin{proposition}
    \label{prp:dequeexists}
  For any $\ell \in \NN$, there exists a 1-push-pop Gray code for length $\ell$.
\end{proposition}

\begin{proof}
  This follows from the existence of \emph{universal words} of all
  sizes~\cite{chung1992universal}, where a \emph{universal word} is a word which contains each word of $\{0,1\}^\ell$ as a factor exactly once. Letting $u = u_1 \cdots u_n$ be such a word, we form the sequence $u_1 \cdots u_{k}, u_2 \cdots u_{k+1}, \ldots, u_{n-k+1} \cdots u_n$. %
\end{proof}

A related result is the existence of a Hamiltonian path in the so-called \emph{shuffle
exchange network}~\cite{feldmann1996shuffle}. This implies the existence of a variant of push-pop quasi-Gray codes, namely, for all $\ell$ we can build a sequence in bijection with $\{0,1\}^\ell$ where each word can be obtained from the previous one by either:
\begin{itemize}
    \item  a circular shift, i.e., one pop on one endpoint followed by a push of the same symbol on the other endpoint; or 
    \item a toggle of the last bit, i.e., one pop on the right endpoint followed by a push of the other symbol on the same endpoint.
\end{itemize}
\end{toappendix}

The section is structured as follows. We first formally define deque machines.
Then, we build a first deque machine~$\calD_0$ which is constant-delay but only \emph{prefix-Hamiltonian}: it produces all words but then loops forever instead of halting. We then investigate how $\calD_0$ can be improved to be Hamiltonian, i.e., to correctly halt when enumeration has finished.
We discuss in Appendix~\ref{apx:dequecounter} how to implement a counter with deque machines by supporting the decrement operation.

\subparagraph{Machine model.}
A \emph{deque machine} $\calM = (Q, \qi, \qh, Q_\OO, \delta)$ consists of a
set $Q$ of \emph{states},
an \emph{initial state} $\qi \in Q$, 
a \emph{halting state} $\qh \in Q$,
a set $Q_\OO \subseteq Q$ of \emph{output states},
and a \emph{transition function}
$\delta \colon (Q \setminus \{\qh\}) \times \bool \times \bool \to Q \times \{\LL, \RR\} \times \{\LL, \RR\} \times \bool$. 
The arguments to the transition function describe the current state and the
symbols at the left and right endpoints of the word, and the return values describe the new
state, the endpoint at which the machine pops a symbol, the endpoint at which
the machine pushes a symbol, and which symbol it pushes. (If both endpoints
are the same, then the machine simply overwrites the symbol at that endpoint.)
Thus, a \emph{configuration} of such a machine consists simply of a state and of a
non-empty word. The semantics of deque machines is again defined as expected: see Appendix~\ref{apx:deque-semantics} for details. 
Like in
\cref{sec:tape}, for a word length $\ell \geq 1$, we start the machine with the
word $0^\ell$ in state $q_\ii$ and the machine produces the current word in unit time whenever it visits an output state. We will investigate machines that are Hamiltonian and constant-delay
in the same sense as in \cref{sec:tape}.

One first idea to design deque machines would be to translate tape machines 
in a black-box fashion by taking the conjugate of the tape at the position of the head. 
Unfortunately, this naive approach does not work as-is: we discuss this in Appendix~\ref{apx:tape2deque}. Thus, the constructions presented in this section require new ideas on top of techniques from the previous section.

\begin{toappendix}
  \subsection{Formal Semantics of Deque-Machines}
  \label{apx:deque-semantics}
In this section, we formally define the semantics of deque machines.
Let $\calM = (Q, \qi, \qh, Q_\OO, \delta)$ be a deque machine.
A \emph{configuration} is simply a pair $c = (q, w)$ of a state $q \in Q$ and a word $w \in \bool^+$. Given $c$, the \emph{next configuration} $c'$ of $\calM$ is
formally defined as follows. Let $a$ and $b$ be the first and last bits of~$w$, respectively, and write $(q', d_1, d_2, b') \coloneq \delta(q, a,
b)$. Let $w'$ be obtained from~$w$ by removing the bit at the endpoint indicated
by~$d_1$, formally $w' \coloneq w[2] \cdots w[n]$ if $d_1 = \LL$
and $w' \coloneq w[1] \cdots w[n-1]$ if $d_1 = \RR$.
Then, let $w''$ be obtained from~$w'$ by
adding the bit $b'$ to the left or right depending on~$d_2$, formally $w''
\coloneq b' w'$ if $d_2 = \LL$ and $w'' \coloneq w' b'$ if $d_2 = \RR$.
Then the \emph{next configuration} is $c = (q', w'')$.

  The \emph{run} of~$\calM$ on~$w$ is then defined like in
  Appendix~\ref{apx:semantics-tape}.
The \emph{initial configuration} of $\calM$ on~$w$ is $c_0 \coloneq
(q_0,w_0)$ with $q_0 = q_\ii$ and $w_0 \coloneq 0^\ell$.
At any step $t \in \NN$ where we reach configuration $c_t = (q_t, w_t)$, if
$q_t$ is an output state ($q_t \in Q_\OO$) then the machine produces (in
unit time) the current word~$w_t$. Further, if $q_t$ is the halting state ($q_t =q_\ff$)
 then the machine halts, otherwise we define 
$c_{t+1}$ as the next configuration of~$\calM$ on~$c_t$.
This defines a sequence
$c_0, c_1, \ldots$ of configurations, which is infinite if the machine does not
halt.
\end{toappendix}

\begin{toappendix}
  \subsection{From Tape Machines to Deque Machines}
  \label{apx:tape2deque}

In this appendix, we remark that it is possible
to translate tape machines 
to deque machines in a systematic way sketched as follows.
When a tape machine $\calT$ stores $\bow u
\omark{x} v \eow$, the corresponding deque machine $\calD$ stores the
conjugate word $x v u$. Then, reads and writes by $\calT$ can be performed by
$\calD$ on the endpoints, and head movements of $\calT$ can be performed by circular shifts in~$\calD$. 

However, the translation of a Hamiltonian tape machine is generally not a Hamiltonian deque machine, for two reasons.
First, the tape machine $\calT$ can detect the markers $\bow$ and $\eow$ while the
deque machine $\calD$ cannot. Second, when $\calT$ produces a word $u \omark{x} v$ then
$\calD$ produces the conjugate word $x v u$.
This conjugation operation is \textbf{not} injective on words due to the position of the head, so $\calD$ might miss some words and have duplicates.

For instance, if we consider the Hamiltonian constant-delay machine $\calT_0$ presented in \cref{prp:t0-correct} in \cref{sec:tape}, then
the
second and third produced words (blue numbers) of Figure~\ref{fig:tape-easy-tree}, namely
$\omark{1}000$ and $00\omark{1}0$, correspond to the same conjugate
word $1000$. Thus, the image of $\calT_0$ by the translation sketched here would not be a prefix-Hamiltonian deque machine.
\end{toappendix}

\begin{figure}[t]
    \centering
    \subfloat[All binary words of length 4 (except 0000), and how they encode a complete binary tree of depth~3.
    Double lines indicate nodes~$\lambda_0$ (1000) and~$\lambda_1$ (1100).
    ]{\label{fig:deque-easy-tree}\begin{tikzpicture}%
    \let\hdist\relax%
    \let\vdist\relax%
    \let\vsep\relax%
    \let\yIII\relax%
    \let\yII\relax%
    \let\yI\relax%
    \newcommand{\upsign}{\triangle}
    \newcommand{\downsign}{\triangledown}
    \newlength{\hdist}\setlength{\hdist}{18mm}%
    \newlength{\vdist}\setlength{\vdist}{7mm}%
    \newlength{\vsep}\setlength{\vsep}{0.5mm}%
    \newlength{\yIII}\setlength{\yIII}{0mm}%
    \newlength{\yII}\setlength{\yII}{\dimexpr\vdist/2\relax}%
    \newlength{\yI}\setlength{\yI}{\dimexpr\vdist*3/2\relax}%
    \newcommand{\localrule}{\rule[-.2em]{0pt}{1.1em}}
    
    \path(0,\vdist*3.5) node[tape] (root) {\localrule\color{black}0001};
    \path (root.south) ++(0pt,-1pt) node[order, anchor=north] {\nth{2}};

    \newif\iffirstleaf
    \foreach \i in {0,1} {
        \firstleaftrue
        \foreach \ii in {0,1} {
            \foreach \iii in {0,1} {
                \path (3*\hdist,\yIII) node[\iffirstleaf double \fi tape] (\i\ii\iii) {\localrule\color{black}1\color{cbdarkorange}\i\ii\iii};            
                \path (\i\ii\iii.east)
                ++(1pt,0pt) node[order,anchor=west] {\nth{\the\numexpr\i*7+\ii*3+\iii+4\relax}}
                ;
                \global\advance \yIII by \dimexpr\vdist\relax
                \global\firstleaffalse
            }
            \path (2*\hdist,\yII) node[tape] (\i\ii) {\localrule\color{black}01\color{cbdarkorange}\i\ii};            
            \path (\i\ii.south) ++(0pt,-1pt) node[order, anchor=north] {\nth{\the\numexpr\i*7+\ii*3+3\relax}}
            ;
            \path[edge] (\i\ii) to node[near start,sloped,below] {0} (\i\ii0);
            \path[edge] (\i\ii) to node[near start,sloped,above] {1}(\i\ii1);
            \global\advance \yII by \dimexpr\vdist*2\relax
        }
        \path (1*\hdist,\yI) node[tape] (\i) {\localrule\color{black}001\color{cbdarkorange}\i};
        \path (\i.north)
            ++(0pt,1pt) node[order,anchor=south] {\nth{\the\numexpr\i*7+9\relax}}
            ;
        \path[edge] (\i) to node[near start,sloped,below] {0} (\i0);
        \path[edge] (\i) to node[near start,sloped,above] {1} (\i1);
        \global\advance \yI by \dimexpr\vdist*4\relax
    }
    \path[edge] (root) to node[near start,sloped,below] {0} (0);
    \path[edge] (root) to node[near start,sloped,above] {1} (1);

\end{tikzpicture}%
}%
    \hfill%
    \subfloat[Definition of the deque machine~$\mathcal{D}_0$, with the transition labels $(\dagger)$ and $(\ast)$ being used in \cref{obs:parity}
    ]{\label{fig:deque-easy-table}
    \begin{minipage}[c]{.45\linewidth}%
    \noindent\input{table/deque-easy}
    \end{minipage}}
  \caption{Illustrations for Proposition~\ref{prp:d0-correct}}
    \label{fig:deque-easy}
\end{figure}

\subparagraph{Constructing prefix-Hamiltonian machines.}
We will start by designing a deque machine which is almost Hamiltonian but fails to halt. Formally, a deque machine is \emph{prefix-Hamiltonian} if it satisfies the following: for
every length $\ell \geq 1$, the $2^\ell$ first words that it produces are all
the words of $\{0,1\}^\ell$ (without duplicates). After this, the machine can continue to run indefinitely, and in
particular it is allowed to produce words that it has already produced.
A Hamiltonian machine is of course prefix-Hamiltonian, but not vice-versa.
The constant-delay requirement is unchanged: in particular it is imposed on the entire
run. 
We now show:

\begin{toappendix}
  \subsection{Proof of \texorpdfstring{\cref{prp:d0-correct}:}{the} Existence of Constant-Delay Prefix-Hamiltonian Deque Machines}
  \label{apx:dequealmost}
\end{toappendix}

\begin{propositionrep}
  \label{prp:d0-correct}
  The deque machine~$\calD_0$, defined in \cref{fig:deque-easy-table}, is constant-delay and prefix-Hamiltonian.
\end{propositionrep}
\begin{proof}
  Let $\ell \geq 1$ be the word length.
  We again handle the zero word $0^\ell$ in a specific way. The other words, called \emph{non-zero} words, can be written
  by distinguishing this time their leftmost $1$, namely, $0^k 1 w$ for some binary word $w$ of length $\ell-k-1$. We again
  use a correspondence between such words and nodes of the binary tree of depth
  $\ell-1$ (but slightly different from the one we used in \cref{prp:t0-correct}), illustrated in
  Figure~\ref{fig:deque-easy-tree}. The root of the tree is the word
  $0^{\ell-1}1$. If $0^k1w$ with $k\geq 1$ is a node $\ell$ of the tree, then
  $0^{k-1}1w0$ is its $0$-child and $0^{k-1}1w1$ is its $1$-child. In summary,
  the word $0^k1w$ encodes the node of the tree that is reachable from the root
  by the path $w$ (reading $w$ from left to right again).

  The deque machine $\calD_0$ then explores this tree with the even-odd trick, very similarly to $\calT_0$ from \cref{prp:d0-correct},
  and again memorizing in the state the parity of the depth of (the node corresponding to) the current word, and whether it is going down or up. More precisely, on a word $w$:
  \begin{itemize}
      \item The machine can detect whether $w$ corresponds to a leaf or not by checking whether the leftmost character is a $1$ or not;
      \item For $b \in \{0, 1\}$, if $w$ does not correspond to a leaf, the machine can go to the $0$-child of (the node corresponding to) $w$ by popping to the left and pushing a $b$ to the right;
      \item The machine can go to the sibling of (the node corresponding to) $w$, by flipping the rightmost bit;
      \item If $w$ does not correspond to the root, the machine can go to the parent by popping to the right and pushing a $0$ to the left.
  \end{itemize}
  The
  transitions of this deque machine and the corresponding tree for $\ell=4$ are
  shown in Figure~\ref{fig:deque-easy-tree} (with the part of the word corresponding to the path in the tree
  colored in orange).
  
  Notice that, unlike $\calT_0$ from \cref{prp:t0-correct}, the machine cannot detect whether the current word corresponds to the root or not: the rest of \cref{sec:deque} will address this issue in detail.
\end{proof}
\begin{proofsketch}
Just like the tape machine~$\calT_0$ (\cref{fig:tape-easy-table}) from~\cref{sec:tape}, the deque machine~$\calD_0$ is based on the depth-first traversal of a complete binary tree.
Once again~$0^\ell$ is treated independently and a word with one or more 1's encodes a node in the complete binary tree of depth~$\ell-1$.
However, the encoding is slightly different%
\footnote{In the
tape machine $D_0$, recall that the tape $w\omark{1}0^k$
encodes the node reached by reading $w$ from the root (from left to right),
and that we modify~$w$ only on the right.
On the other hand, a deque machine can only modify the endpoints of its deque. 
Hence the same node needs to be encoded either with $0^k1w$ or $\overleftarrow{w}10^k$, where~$\overleftarrow{w}$ is the mirror of~$w$. 
We chose the former encoding in order for the path from the root to the node to be consistently read from left to right.}:
the tape machine $\calT_0$ from the previous section distinguished the \emph{rightmost} 1 and 
the node reachable from the root by a path labeled $w$ was encoded by~$w\omark{1}0^{\ell-|w|-1}$, but the deque machine $\calD_0$ will instead distinguish the \emph{leftmost} 1 
and encode that node by~$0^{\ell-|w|-1}1w$.
The proof hinges on the analogue of \cref{clm:t0} for~$\calT_0$:
\begin{claim}\label{clm:d0-correct}
  During the enumeration of the first $2^\ell$ words, $\calD_0$ visits each deque of the form $0^k1\hlpath{w}$ exactly twice, first in state $(\downarrow,p)$, then later in state $(\uparrow,p)$, where $p = |w| \bmod 2$.
\end{claim}
See \cref{fig:deque-easy-tree} for the full binary tree in the case~$\ell=4$ and
Appendix~\ref{apx:dequealmost} for more details. 
\end{proofsketch}

The machine $\calD_0$ is \emph{not} Hamiltonian because it does not halt at the correct time, due to the lack of left and right markers ($\bow$, $\eow$).
Indeed, $\calD_0$ should halt when it visits the root the second time (i.e., in the configuration~$(\upstate,\Even)~0^{\ell-1}1$). 
We did not find how to detect that the current deque is~$0^{\ell-1}1$:
e.g., $\calD_0$ cannot store the number of $1$'s in the current word, or the current number of leading $0$'s, because these numbers depend on $\ell$ so they are not constant. This motivates the more complex techniques presented in the rest of this section.

To present them, we will need
to understand the behavior of $\calD_0$ when it reaches the root for the second time, that is, in state $(\upstate,\Even)$. This corresponds to the step where the machine has produced every word and ``should'' halt.

\begin{observation}[Parity is swapped at the root]
\label{obs:parity}
When the run of~$\calD_0$ visits the root (i.e., the word~$0^{\ell-1}1$) at the end of the traversal, it is with state~$(\upstate,\Even)$.
Then, the run continues as follows. (Transitions $(\ast)$ and $(\dagger)$, below, are the ones labeled in \cref{fig:deque-easy-table}.)
\begin{equation*}%
\cdots \quad\leadsto\quad (\upstate,\Even)~0^{\ell-1}1 
\quad\leadsto[(\ast)]\quad (\upstate,\Odd)~0^\ell
\quad\leadsto[(\dagger)]\quad (\downstate,\Odd)~0^{\ell-1}1
\quad\leadsto\quad\cdots
\end{equation*}

At this point, the machine performs a second depth-first traversal of the tree.
Recall that the first traversal starts in configuration $((\downarrow, \Even),~0^{\ell-1}1)$ while we can see above that the second starts in $((\downarrow, \Odd),~0^{\ell-1}1)$.  This means the parity will be swapped for the full duration of the second traversal: during the second traversal, a claim similar to \cref{clm:d0-correct} holds with states $(\downstate,\bar p)$ and $(\upstate,\bar p)$.
When the second depth-first traversal ends, that is when $\calD_0$ reaches the root for the fourth time, it is with state $(\upstate,\Odd)$.
Then, the parity is swapped again by $(\ast)+(\dagger)$ and $\calD_0$ restarts, looping indefinitely.
\end{observation}

\subparagraph{Constructing constant-delay Hamiltonian deque machines.}
In the following, we improve~$\mathcal{D}_0$ (\cref{fig:deque-easy-table}) in order to make it Hamiltonian, i.e., to detect when the enumeration can halt. This hinges on the fact that, unlike the root, some nodes of the complete binary tree of depth $\ell-1$ can be detected. Let $\lambda_0$ be the first leaf of the 0-subtree of the root (corresponding to the word $100^{\ell-2}$), and $\lambda_1$ the first leaf of the 1-subtree (corresponding to the word $110^{\ell-2}$). The nodes $\lambda_0$ and $\lambda_1$ are double-circled in \cref{fig:deque-easy-tree}. 
Then:

\begin{observation}[$\lambda_0$ and $\lambda_1$ can be detected]
\label{obs:100*-detection}
During the run of~$\calD_0$, the first time that $\lambda_0$ is visited is the first time in the run that $10$ appears to the left of the current word, and the first time $\lambda_1$ is visited is the first time in the run that $11$ appears to the left.
\end{observation}

Note that, within the constant-delay bound of our deque machines, at each step we can pop the two leftmost characters and push them back; so indeed our machines can detect whenever the words $10$ or $11$ appear for the first time at the left endpoint. 
Our future deque machines will leverage \cref{obs:100*-detection} as follows.
The traversal starts at $\lambda_0$ (i.e., the word $100^{\ell-2}$) instead of the root~$r$.
Then, they count the number of \emph{half-traversals} since the start.
One half-traversal spans either from $\lambda_0$ (i.e., $100^{\ell-2}$) to $\lambda_1$ (i.e.,  $110^{\ell-2}$), or conversely.
The former visits the 0-subtree, minus the descent in the branch~$r\to\lambda_0$, plus the descent in the branch~$r\to\lambda_1$.
The latter visits the remainder of the 1-subtree, plus the descent in the branch~$r\to\lambda_0$.
Unfortunately, we cannot simply build a machine that performs two half-traversals, because 
by \cref{obs:parity} the parity is swapped when the root is visited, so that the incorrect parity during the descent 
in the branch $r\to\lambda_0$ would lead to duplicates.
We will now show two alternative ways to use the above and modify $\calD_0$ to a Hamiltonian machine: the \emph{double-traversal} technique, and the \emph{lookahead} technique.

\subparagraph{Double-traversal technique.}
We first build a Hamiltonian machine $\calD_1$ from $\calD_0$ 
with the double-traversal technique.
The machine $\calD_1$ will perform \emph{four} half-traversals (i.e., two traversals) starting from $\lambda_0$, and then it will halt. Each of these two traversals has an opposite parity by \cref{obs:parity}. We build the machine to ensure that the two traversals partition the output (using the parity), even though the machine is not aware of the moment when the parity is swapped.

\begin{figure}[ht]
    \input{table/deque-double}
    \caption{Definition of the constant-delay Hamiltonian deque machine~$\calD_1$}
    
    \label{fig:deque-double-table}%
\end{figure}

\begin{proposition}
    \label{prp:correct-d1}
    The deque machine~$\calD_1$ of \cref{fig:deque-double-table} is constant-delay and Hamiltonian.
\end{proposition}

We sketch the proof of \cref{prp:correct-d1} in the rest of this paragraph, before turning to the alternative \emph{lookahead} technique.
Let~$\ell$ be the word length. We ignore the two words~$0^{\ell-1}1$ and~$0^{\ell-1}0$, that are handled in a special way (see \cref{clm:d1-special-words}).
The machine $\calD_1$ traverses the complete binary tree $B_{\ell-2}$ of depth $\ell-2$ (in contrast with $\calD_0$, which used $B_{\ell-1}$; we will explain this later).
Like $\calD_0$, the machine $\calD_1$ 
stores whether it is currently going up or down the tree, along with a parity bit in the second component of its internal state.     
However, unlike $\calD_0$, the machine $\calD_1$ starts at $\lambda_0$ (the first leaf of the 0-subtree of the root, corresponding to the word $100^{\ell-3}$).
For this reason, the parity bit maintains the parity \emph{of the height} instead of the depth:
the \emph{correct parity} of a node encoded by $0^{k}1w$ is~$|k|\bmod2$. (In particular, the correct parity of $\lambda_0$ is $0$, represented as $\Even$ in the state of the machine.)

Moreover~$\calD_1$ maintains in its internal state the index of the current half-traversal ($\{\nth{1},\nth{2},\nth{3},\nth{4}\}$): the machine stops after reaching $\lambda_0$ in the \nth{4} half-traversal. The four half-traversals are regrouped into a \emph{first full traversal} and a \emph{second full traversal}, each going from $\lambda_0$ back to $\lambda_0$.
Let us further split each full traversal in two parts called A and B:
\begin{itemize}
    \item Part A is the majority of the traversal: starting at  $\lambda_0$, traversing the tree until the root ($0^{\ell-2}1$) is reached.
    \item Part B corresponds to the last descent from the root ($0^{\ell-2}1$) down to $\lambda_0$ again, hence passing through all nodes of the form~$0^k10^{\ell-k-2}$.
\end{itemize}
Parts A and B are \emph{not} aligned with half-traversals, though part B is always included in the second half of each traversal (i.e., the visit of $\lambda_1$ happens in part~A). 
Let us stress that while $\calD_1$ maintains the current half-traversal, it does not know 
whether it is in part A or part B of the current full traversal.
Furthermore, due to \cref{obs:parity}, the parity bit is equal to the correct parity during part A of the first traversal and it is swapped during part B of the first traversal.
On the other hand, the parity bit is swapped during part A of the second traversal and correct again during part B of the second traversal. 

An overview of the run of~$\calD_1$, with half-traversal numbers and parity bit status (correct/swapped), is given in \cref{fig:deque-double-explanation} of the appendix. 
Altogether, each node is visited once in each of the four possible combinations of ${\{\upstate,\downstate\} \times \{\Even,\Odd\}}$.
In order to produce each word only once, $\calD_1$ writes the current parity bit as the rightmost symbol of the current deque (with the convention $\Even=0$ and $\Odd=1$).
This special treatment of the rightmost symbol of the deque is why $\calD_1$ performs traversals of the binary tree~$B_{\ell-2}$ instead of~$B_{\ell-1}$.
(We note that, even though the parity bit is always written as the rightmost symbol of the deque, we still need to additionally keep track of parity in the internal state of~$\calD_1$, because the parity is part of the definition of which states are output states.)

The proof of \cref{prp:correct-d1} consists in showing that~$\calD_1$ has the invariants stated by Claims~\ref{clm:d1-correct} and~\ref{clm:d1-special-words}, given below.
They are established by a simple inspection of the transition table (\cref{fig:deque-double-table}).

\begin{claim}\label{clm:d1-correct}
    Let $v\coloneq 0^{k}1\hlpath{u}$
    be a word of length~$\ell-1$, for some~$k\in\{0,\ldots,\ell-2\},u\in\{0,1\}^*$.
    Let $p=|k|\bmod 2$ be the correct parity, and define $w \coloneq v \hlpari{p}$ and $\overline{w} \coloneq v \hlpari{\overline{p}}$.
    The deques $w$ and~$\overline{w}$ are both visited twice by the run of~$\calD_1$, in the states and order given below.
    \begin{itemize}
        \item If~$u=\varepsilon$ (i.e., $w$ is the root): $w$ is visited by states $(\upstate,p,\nth{2})$ and $(\downstate,p,\nth{4})$; while $\overline{w}$ is visited by states $(\downstate,\overline{p},\nth{2})$ and $(\upstate,\overline{p},\nth{4})$.
        \item Otherwise, if~$u\in 0^*$ (i.e.,~$w$ is in the first branch of the 0-subtree)): 
            $w$ is visited by states $(\upstate,p,\nth{1})$ and $(\downstate,p,\nth{4})$; while $\overline{w}$ is visited by states $(\downstate,\overline{p},\nth{2})$ and $(\upstate,\overline{p},\nth{3})$.
        \item Otherwise, if~$u\in 10^*$ (i.e.,~$w$ is in the first branch of the 1-subtree): 
            $w$ is visited by states $(\downstate,p,\nth{1})$ and $(\upstate,p,\nth{2})$; while $\overline{w}$ is visited by states $(\downstate,\overline{p},\nth{3})$ and $(\upstate,\overline{p},\nth{4})$.
        \item Otherwise: $w$ is visited by states $(\downstate,p,s)$ and $(\upstate,p,s)$;
        while  $\overline{w}$ is visited by $(\downstate,\overline{p},s')$ and $(\upstate,\overline{p},s')$;
        where~$(s,s')=(\nth{1},\nth{3})$ if the first bit of~$u$ is~$0$, or~$(s,s')=(\nth{2},\nth{4})$ otherwise.
    \end{itemize}
\end{claim}

\begin{claim}\label{clm:d1-special-words}
    (Beware that the $1$ below is a parity bit.)
    \begin{itemize}
        \item The deque~$w=0^{\ell-1}\hlpari{1}$ is visited by state $(\upstate,\Odd,\nth{2})$
    if~$\ell$ is even, or by state $(\upstate,\Odd,\nth{4})$ otherwise.
        In particular,~$w$ is produced once during the traversal.
        \item The deque~$\overline{w}=0^{\ell-1}\hlpari{0}$ is visited by state~$q_i$ and then by state~$(\upstate,\Even,\nth{4})$ if~$\ell$ is even, or then by~$(\upstate,\Even,\nth{2})$ otherwise.
        In particular,~$\overline{w}$ is produced once during the traversal.
    \end{itemize}
\end{claim}

\begin{toappendix}
\begin{figure}
    \centering
    \begin{tikzpicture}%
    \let\hdist\relax\newlength{\hdist}\setlength{\hdist}{12mm}%
    \let\vdist\relax\newlength{\vdist}\setlength{\vdist}{12mm}%
    \let\vsep\relax\newlength{\vsep}\setlength{\vsep}{6mm}%
    \let\vsepii\relax\newlength{\vsepii}\setlength{\vsepii}{9mm}%
    \let\vsepdiff\relax\newlength{\vsepdiff}\setlength{\vsepdiff}{\dimexpr\vsepii-\vsep\relax}%
    \tikzset{
    annot/.style={fill=white,draw,line width=.2mm,solid,circle,font=\footnotesize,minimum width=2mm,inner sep=1pt,pos=\annotpos},
        arrow/.style={sloped, pos=\arrowpos, allow upside down},
    run/.style={draw,line width=.4mm,>={Stealth[length=2.5mm]},line cap=round},    
    }
    \newcommand{\annotpos}{.5}
    \newcommand{\arrowpos}{.5}

    \path(0,0) node[draw=black,vertex,fill=cbyellow] (root) {};
    \path ($(root)+(-2\hdist,0)$) node[draw=black,vertex,fill=cbyellow] (pre-root) {};
    \path[draw] ($(root)+(4\hdist,-2\vdist)$) coordinate (south-subtree-root)
        -- ++(2\hdist,-\vdist) node[draw=black,vertex,fill=cbyellow] (south-subtree-first-leaf){}
        -- ++(0,2\vdist) coordinate (south-subtree-last-leaf)
        -- cycle;
    \path[draw] ($(root)+(4\hdist,2\vdist)$) coordinate (north-subtree-root)
        -- ++(2\hdist,\vdist) coordinate (north-subtree-last-leaf)
        -- ++(0,-2\vdist) node[draw=black,vertex,fill=cbyellow] (north-subtree-first-leaf) {}
        -- cycle;
    \path[draw,->, shorten >=2pt] (root) to node[below left] {} (south-subtree-root);
    \path[draw,->, shorten >=2pt] (root) to node[above left] {} (north-subtree-root);

    \foreach \i/\sep in {1/\vsep,2/\vsepii} {
        \foreach \s/\m/\e in {%
            south-subtree-root/south-subtree-first-leaf/south-subtree-last-leaf,%
            south-subtree-first-leaf/south-subtree-last-leaf/south-subtree-root,%
            north-subtree-root/north-subtree-first-leaf/north-subtree-last-leaf,%
            north-subtree-root/north-subtree-first-leaf/north-subtree-last-leaf,%
            north-subtree-first-leaf/north-subtree-last-leaf/north-subtree-root,%
            root/pre-root/root%
            }
        {
            \path (\s) -- (\m) -- ([turn]-90:\sep) coordinate (\m-start-\i);
            \path (\e) -- (\m) -- ([turn]90:\sep) coordinate (\m-end-\i);
        
        }
    }

    \newcommand{\arrowIn}{
    \tikz \draw[run,solid,->] (-2.5mm,0) -- (1pt,0);
    }
    \newcommand{\curstate}{}

    \renewcommand{\annotpos}{.2}
    \renewcommand{\arrowpos}{.6}
    \renewcommand{\curstate}{1}
    \path[run,cborange,->,rounded corners=0]
    ($(south-subtree-first-leaf)+(0:\vsep)$)
        node[black,anchor=north, inner sep=.5mm,xshift=-1mm]{\tiny start}
        to
        node[annot]{\curstate}
        node[arrow]{\arrowIn}
    ($(south-subtree-last-leaf-start-1)$)
    arc[start angle=0, end angle=90+31.75, radius=\vsep]
        to[rounded corners] 
        node[annot]{\curstate}
        node[arrow]{\arrowIn}
    ($(south-subtree-root)+(90:\vsep)$)
        to[rounded corners] 
        node[annot]{\curstate}
        node[arrow]{\arrowIn}
    ($(root)+(0:1.66\vsep)$)
        to[rounded corners] 
        node[annot]{\curstate}
        node[arrow]{\arrowIn}
    ($(north-subtree-root)+(-90:1\vsep)$)
        to
        node[annot]{\curstate}
    ($(north-subtree-first-leaf)+(180+63:1\vsep)$)
        arc[start angle=180+63.5,end angle=270+31.75,radius=\vsep]
    ;

    \renewcommand{\curstate}{2}
    \path[run,cbblue,-]
    ($(north-subtree-first-leaf)+(270+31.75:1\vsep)$)
        arc[start angle=270+31.75,end angle=360,radius=\vsep]
        to
        node[annot]{\curstate}
        node[arrow]{\arrowIn}
    ($(north-subtree-last-leaf-start-1)$)
        arc[start angle=0, end angle=180-63.5, radius=\vsep]
        to[rounded corners]
        node[annot]{\curstate}
        node[arrow]{\arrowIn}
    ($(root)+(90:1\vsep)$)
        to 
        node[annot]{\curstate}
        node[arrow]{\arrowIn}
    ($(pre-root)+(0,\vsep)$)
    arc[start angle=90, end angle=180, radius=\vsep]
    ;
    \path[run,cbblue,->,dashed]
    ($(pre-root)+(-\vsep,0)$)
    arc[start angle=180, end angle=270, radius=\vsep]
        to
        node[annot]{\curstate}
    ($(pre-root)+(1\vsep,-1\vsep)$)
        to[rounded corners]
        node[arrow]{\arrowIn}
    ($(root)+(-90:\vsep)$)
        to
        node[arrow]{\arrowIn}
        node[annot]{\curstate}
        node[annot,pos=.9]{\curstate}
    ($(south-subtree-first-leaf-start-1)$)
        to
        ([turn]0:\vsepdiff*1.13)%
        arc[start angle=180+63.5,end angle=270+31.75,radius=\vsep]
        coordinate (end2)
    ;

    \renewcommand{\annotpos}{.8}
    \renewcommand{\arrowpos}{.4}
    \renewcommand{\curstate}{3}
    \path[run,cbteal,dashed,->]
        (end2)
        arc[start angle=270+31.75, end angle=360, radius=6mm]
        to
        node[annot]{\curstate}
        node[arrow]{\arrowIn}
        ($(south-subtree-last-leaf-start-2)$)
        arc[start angle=0, end angle=90+31.75, radius=\vsepii]
        to[rounded corners] 
        node[annot]{\curstate}
        node[arrow]{\arrowIn}
        ($(south-subtree-root)+(90:\vsepii)$)
        to[rounded corners] 
        node[annot]{\curstate}
        node[arrow]{\arrowIn}
        ($(root)+(0:1.66\vsepii)$)
        to[rounded corners] 
        node[annot]{\curstate}
        node[arrow]{\arrowIn}
        ($(north-subtree-root)+(-90:\vsepii)$)
        to[rounded corners] 
        node[annot]{\curstate}
        node[arrow]{\arrowIn}
        ($(north-subtree-first-leaf-start-2)$)
        arc[start angle=180+63.5, end angle=270+31.750, radius=\vsepii];

    \renewcommand{\curstate}{4}
    \path[run,cbpink,dashed]
    ($(north-subtree-first-leaf)+(270+31.750:\vsepii)$)
        arc[start angle=270+31.750, end angle=360, radius=\vsepii]
        to[rounded corners]
        node[annot]{\curstate}
        node[arrow]{\arrowIn}
    ($(north-subtree-last-leaf-start-2)$)
        arc[start angle=0, end angle=180-63.5, radius=\vsepii]
        to[rounded corners]
        node[annot]{\curstate}
        node[arrow]{\arrowIn}
    ($(root)+(0,\vsepii)$)
        to
        node[arrow]{\arrowIn}
        node[annot]{\curstate}             
    ($(pre-root)+(0,\vsepii)$)
    arc[start angle=90, end angle=180, radius=\vsepii]
    ;
    \path[run,cbpink,->]    
    ($(pre-root)+(-\vsepii,0)$)
    arc[start angle=180, end angle=270, radius=\vsepii]
        to[rounded corners]
        node[annot]{\curstate}
        node[arrow]{\arrowIn}
    ($(root)+(0,-\vsepii)$)
        to[rounded corners]
        node[arrow]{\arrowIn}
        node[annot]{\curstate}
    ($(south-subtree-first-leaf-start-2)$)
    node[black,anchor=north, inner sep=.5mm,xshift=-1mm]{\tiny end};

    \path($(pre-root)+(5mm,1mm)$) node[anchor=south west,inner sep=.5mm] (pre-root-lab) {\small$z$};  
    \path[gray,draw,-] (pre-root) to[out=10,in=180] (pre-root-lab); 
    
    \path($(root)+(-5mm,-1mm)$) node[anchor=north east,inner sep=.5mm] (root-lab) {\small$r$};
    \path[gray,draw,-] (root) to[out=190,in=0] (root-lab); 

    \path($(south-subtree-first-leaf)+(0mm,5mm)$) node[anchor=south east](south-subtree-first-leaf-lab) {\small$\lambda_0$};  
    \path[gray,draw,-] (south-subtree-first-leaf-lab) to[out=-90,in=115] (south-subtree-first-leaf);   
    
    \path($(north-subtree-first-leaf)+(0mm,5mm)$) node[anchor=south east](north-subtree-first-leaf-lab) {\small$\lambda_1$};
    \path[gray,draw,-] (north-subtree-first-leaf-lab) to[out=-90,in=115] (north-subtree-first-leaf);   
    
\end{tikzpicture}%
    \caption{Explanation of the double-traversal technique in the proof of \cref{prp:correct-d1}. 
    The complete binary tree of depth $\ell-2$ is drawn in black and three nodes of particular interest are indicated: 
    the root~$r$, where the run deviates from a depth-first traversal;
    the first leaves $\lambda_0$ and $\lambda_1$ in the 0-subtree
    and the 1-subtree of the root, respectively (which we can detect during the run, and from which the traversals start and end).
    In addition, a false node $z$ is represented left of the root; it corresponds to the zero-deque and is visited during the run at the moment when parity is swapped by \cref{obs:parity}. 
    The four half-traversals (called $\{\nth{1},\nth{2},\nth{3},\nth{4}\}$ in the third component of the internal state of $\calD_1$) are given by the color and the circled number. The moments when the parity is ``correct'' (i.e., Part A of the first full traversal and part B of the second full traversal) are shown with solid lines, and the moments when it is swapped (i.e., Part B of the first full traversal and Part A of the second full traversal) are shown with dashed lines.
    }
    \label{fig:deque-double-explanation}
\end{figure}
\end{toappendix}

\subparagraph{Lookahead technique.} An alternative way to make~$\calD_0$ Hamiltonian is to use the \emph{lookahead} technique. We think it is more flexible though more technical than double traversal.

The deque machine $\calD_2$, defined in \cref{fig:deque-lookahead-table}, implements this technique.
More precisely,~$\calD_2$ performs two half-traversals (i.e., one full traversal) of~$B_{\ell-1}$, starting from $\lambda_0$. 
However, each time the machine visits a $1$-child~$n$ and moves to the state~$\downstate$, it performs a “lookahead”: it descends to the leftmost leaf of~$n$ and then ascends back to~$n$, before continuing the traversal (and descending again to the leftmost leaf of~$n$). This first upward return to $n$ is possible because $\calD_2$ can detect the first node that is a $1$-child while moving upward in the tree.

This lookahead phase requires keeping track of the height of the current node modulo~$4$, since in $\calD_2$ each node is visited four times instead of twice as in $\calD_0$. Concretely, a node of height~$h \bmod 4$ is output on its $h^{\text{th}}$ visit.
Like~$\calD_1$, the machine $\calD_2$ detects the second visit of~$\lambda_0$. This event occurs during the lookahead initiated from the root~$r$.  Hence,~$\calD_2$ is able to detect that it is visiting the root after this lookahead phase.  The machine is thus able to adjust the stored height modulo~$4$ before performing the final descent in $r \to \lambda_0$.  This allows to fix the issue from \cref{obs:parity} and to produce the missing solutions in this branch.

\begin{figure}[ht]
    \input{table/deque-lookahead}
    \caption{Definition of the constant-delay Hamiltonian deque
    machine~$\calD_2$. The transition $(\diamond)$ is the simplified version of
  the machine, and is in fact replaced by the three transitions given at the
bottom of the table}
    \label{fig:deque-lookahead-table}
\end{figure}

\begin{proposition}[Lookahead]
    \label{prp:correct-d2}
    The deque machine~$\calD_2$, defined in \cref{fig:deque-lookahead-table}, is constant-delay and Hamiltonian.
\end{proposition}

We sketch the proof of Proposition~\ref{prp:correct-d2} in the remainder of the paragraph.
Let~$\ell$ be a word length. 
Once again, we ignore for now the word~$0^{\ell}$ which is handled in a special way, and focus on the other \emph{non-zero} words. The top part of table \cref{fig:deque-lookahead-table} describes a simplified version of the machine, that the reader should consider first. The actual machine is obtained by replacing the transition marked with $(\diamond)$ by the three transitions given at the bottom of \cref{fig:deque-lookahead-table}.  We will explain this at the end of the proof, together with the purpose of the state $\sspecial$.

The deque machine $\calD_2$ traverses the complete binary tree $B_{\ell-1}$ of depth $\ell-1$. We use the same bijection
between nodes of $B_{\ell-1}$ and nonzero words of length~$\ell$. The state of $\calD_2$ contains three components:
\begin{itemize}
\item The first component of the state of $\calD_2$ stores whether it is currently going up or down the tree with four different values $\{\downarrow_0,\uparrow_1,\downarrow_2,\uparrow_3\}$. The intuition is that $\downarrow_0, \uparrow_1, \downarrow_2$ correspond to the value $\downarrow$ of~$\calD_0$, whereas $\uparrow_3$ corresponds to the value $\uparrow$ of $\calD_0$.
\item The second component of the state of $\calD_2$ stores the height modulo 4 of the current node (possibly with an offset after visiting the root, as we will explain).
\item The third component of the state is in $\{\nth{1},\nth{2},\last\}$, which indicates the half-traversal that we are currently doing, as we explain below. 
\end{itemize}
More precisely, for the third component, we want to use \cref{obs:100*-detection} to detect the end of the enumeration.
The initial value used for the third component is $\nth{1}$, which corresponds to the first half-traversal, i.e., the traversal of $B_{\ell-1}$ between $\lambda_0$ (the first leaf of the 0-subtree of the root, corresponding to the word $100^{\ell-2}$) and $\lambda_1$ (the first leaf of the 1-subtree of the root, corresponding to the word $110^{\ell-2}$). When $\calD_2$ traverses the leaf $\lambda_1$, it detects it by \cref{obs:100*-detection}, and then transition ($\dagger$) of \cref{fig:deque-lookahead-table} is used to change the state to $\nth{2}$, indicating that the second half-traversal has started. 
Finally, when $\lambda_0$ is traversed a second time by $\calD_2$, we detect it again by \cref{obs:100*-detection}, and then transition ($\ast$) is used, which changes the state to $\last$. The machine then enumerates some missing words of the form $0^* 1 0^*$ (as we will explain at the end), and then halts.

We now explain the traversal in more detail. The machine $\calD_2$ starts at $\lambda_0$ and mimic the traversal of $B_{\ell -1}$ by $\calD_0$. The only difference is that it does a ``lookahead'' when it starts to go down the tree from a node $n$ which is a $1$-child. More precisely, letting $\lambda_n$ be the leftmost leaf of the subtree rooted in $n$, the machine $\calD_2$ traverses the path $n \rightarrow \lambda_n$ three times: first going down (state $\downarrow_0$), then up (state $\uparrow_1$) and finally down (state $\downarrow_2$). When the leaf $\lambda_n$ is reached in state $\downarrow_2$, $\calD_2$ transitions to state $\uparrow_3$, which corresponds to the state $\uparrow$ in $\calD_0$. To slightly simplify the machine, we maintain the opposite of the height (modulo 4), that is: when the machine goes up the height is decremented, and it is incremented when the machine goes down. 

The correctness of~$\calD_2$ rely on Claims~\ref{clm:nonleftpath}
and~\ref{clm:leftpath}, which describe the configurations of~$D_2$ when the tape contains at least two 0's, and exactly one 0, respectively.
Both are proved with straightforward inductions based on Table (\cref{fig:deque-lookahead-table}).
The case of the tape~$0^\ell$ is treated separately.

\begin{claim}
\label{clm:nonleftpath}
Let $w \in \{0,1\}^{\ell}$ be a word which is a non-zero word and \textbf{not} of the form $0^*10^*$ (i.e., $w$ encodes a node that is not on the branch from the root to~$\lambda_0$).
Let $h$ be the opposite of the height of $w$ in $B_{\ell-1}$ modulo 4.
Then, four configurations of the form $(w,q)$ are reached by $\calD_2$, successively in states $q=(\downarrow_0,h,s)$, $q=(\uparrow_1,h,s)$, $q=(\downarrow_2,h,s)$, and $q=(\uparrow_3,h,s)$, where $s$ is either $\nth{1}$ 
($w$ is in between $\lambda_0$ and $\lambda_1$ in the traversal order) or $\nth{2}$ ($w$ is in between $\lambda_1$ and the root). In particular each such word $w$ is visited in exactly one of the output states $\{(\downarrow_0,0,s), (\uparrow_1,1,s), (\downarrow_2,2,s),(\uparrow_3,3,s)\}$, depending on its height.
\end{claim}

The previous claim guarantees that the words are produced once, except for the words of the form $0^* 1 0^*$: these are the words of the \emph{leftmost path}, i.e., the path from the root (corresponding to the word $0^{\ell -1}$) to the leftmost leaf $\lambda_0$ (corresponding to the word $10^{\ell -1}$). We deal with the elements of this branch differently, since we know that their height is shifted by one when going over the root, by \cref{obs:parity} (generalized to the present setting, where we keep track of the opposite of the height, and we do it modulo 4 instead of modulo~2). These words are first traversed when going up from a $1$-child, that is in state $\uparrow_3$, before the traversal has reached the root for the first time --
hence with the correct height. Then, thanks to the lookahead, $\calD_2$ first traverses the leftmost branch up and down in states $\downarrow_0$ and $\uparrow_1$, with a height shifted by one. When again at the root of the tree, it has transitioned to the state $\last$ thanks to the lookahead, with transition ($\ast$): this uses the analogue of \cref{obs:100*-detection} to argue that the machine detects that it has now revisited $\lambda_0$ after having visited $\lambda_1$.
Hence, before $\calD_2$ does the last descent along the leftmost path in state $\downarrow_2$, the height is again decremented by one to avoid duplicates, in transition $\diamond$. 
Hence, all the words of the leftmost path are enumerated exactly once, as precisely stated in the following claim.

\begin{claim}
\label{clm:leftpath}
Let $w = 0^k10^{\ell -1 -k}$ be a word for some~$k<\ell$, and let $h = -k \bmod 4$. 
Then,~$\calD_2$ reaches a configuration of the form $(w,q)$ four times by $\calD_2$, successively in states $q=(\uparrow_3,h,s)$, $q=(\downarrow_0,h-1,\nth{2})$,
$q=(\uparrow_1,h-1,\last)$, and $q=(\downarrow_2,h-2,\last)$, for some value of $s \in \{\nth{1},\nth{2}\}$ (namely $s = \nth{2}$ when $w$ is the root and $s = \nth{1}$ otherwise).

Moreover, the word $w$ is produced exactly once, depending on the value.
\end{claim}

Finally, we must explain how to deal with the word $0^\ell$, which is the purpose of the state called ``$\sspecial$'' and of the three transitions at the bottom of \cref{fig:deque-lookahead-table} to replace transition~$(\diamond)$.
Observe that the machine $\calD_2$ visits the word $0^\ell$ only once in
state $(\uparrow_3,h,\nth{2})$, see \cref{obs:parity}. Hence, this word is produced by $\calD_2$ if and only if $h=3$.
The complete version of the transition table of \cref{fig:deque-lookahead-table} ensures the following: when it visits the root in state $(\uparrow_1,h,\last)$ before the last descent, it first produces $0^{\ell}$ if $h\neq 3$. 
These are the transitions given at the bottom of \cref{fig:deque-lookahead-table} to replace the transition $(\diamond)$: we distinguish two cases depending on whether $h= 3$ (where we do the analogue of $(\diamond)$) or not (where we additionally output $0^\ell$ via the $\sspecial$ state).

We last argue why the machine satisfies the constant-delay requirement. For this, we observe that the second component of the state always stores the opposite of the current height modulo 4 plus some offset $b$, and the offset changes at only constantly many points in the enumeration (namely, the first time we visit the root where it is decremented, and the second time when it is fixed when revisiting the root after the lookahead). So it suffices to show the constant-delay guarantee whenever the offset does not change. We can also ignore the additional states mentioned at the end of the previous paragraph, as they are only visited once.

When the offset does not change, the machine simply performs a traversal of the complete binary tree $B_{\ell -1}$, except that descents from a $1$-child are replaced by a back-and-forth successively involving the states $\downarrow_0, \uparrow_1, \downarrow_2$. If we stay in the same state for more than 3 steps, then we visit all values of the second component, and one of the states encountered must be an output state. So a long run without outputs must change states among $\downarrow_0, \uparrow_1, \downarrow_2, \uparrow_3$ every 3 steps at most. Looking at the possible transitions between states, these state transitions must include a transition from $\downarrow_0$ to $\uparrow_1$, or from $\downarrow_2$ to $\uparrow_3$. But these transitions only happen at leaves of $B_{\ell-1}$, and since we assumed that we do not stay in the same state for more than 3 steps it easy to see that the time between the 4 visits of the leaf that are prescribed by \cref{clm:nonleftpath} must happen within a constant time bound. 

\begin{toappendix}
\subsection{Counter Implementation for Deque Machines}
\label{apx:dequecounter}

In this appendix, we sketch how the deque machines that we presented can be
extended to support a decrement operation, to give a constant-time and
constant additional memory implementation of a counter in the deque machine
model.

The deque machines $\calD_1$ and $\calD_2$ have transitions which are not injective, hence decrement cannot be implemented as easily as in the case of tape machines. 
We can partially fix the lack of injectivity, which comes from maintaining whether we are in the first or the second half traversal. When incrementing,  we detect the transition from the first to the second half traversal 
at node $\lambda_0$ (the leftmost leaf of the right subtree), see \cref{obs:100*-detection}. When decrementing,
we could detect the transition from the second to the first half traversal but only at the rightmost leaf of the left subtree. Fortunately, in between these two leaves, both machines do exactly the same transitions no matter whether their internal state indicates that their are in the first or second half traversal. In particular, each machine enumerates the same words during that time for these two possible internal states.

However, the transitions from the initial state $q_\ii$ are not injective either, i.e., we do not know when we should stop decrementing and move back to $q_\ii$.
This issue is similar to $\calD_0$ being unable to detect the root, and we are not sure whether it can be addressed.
All in all, $\calD_1$ and $\calD_2$ can be easily modified to support the decrement operation in constant time, but only subject to the assumption that the counter they maintain stays positive.
\end{toappendix}

\section{Lower Bounds: Queue Machines and Stack Machines}
\label{sec:stackqueue}
We now study variants of deque machines where
we restrict which operations are allowed on the endpoints. 
We consider \emph{queue machines} and \emph{stack machines}, and show that no such machine can be Hamiltonian.

\subparagraph{Queue machines.}
We define queue machines as a variant of deque machines that can only push to the left and pop to the right;
but we still augment them with the ability to read the $\rho$ rightmost symbols of the word, for some constant $\rho>0$. 
For convenience, we assume without loss of generality that the current word $\ell$ is no smaller than~$\rho$.
A \emph{queue machine} $\calM = (Q, \qi, \qh, Q_\OO, \rho, \delta)$ is then defined similarly to a deque machine,
with the transition function having signature
$\delta \colon (Q \setminus \{\qh\}) \times \{0,1\}^\rho \to Q \times \{0,1\}^{\leq \rho}$. 
The arguments to the transition function describe the current state and the $\rho$ rightmost bits of the word. The return values describe the new state and the word $w$ with length $0 \leq |w| \leq \rho$ pushed to the left of the deque, with the $|w|$ rightmost characters being popped so that 
the word length stays the same.
Note that the definition of queue machines does not allow them to read the leftmost bits of the deque, but they can be modified to do so without loss of generality, e.g., by remembering the~$\rho$ most recently pushed bits as part of their internal state.

We investigate queue machines that are \emph{Hamiltonian} and \emph{constant-delay} in the same sense as in~\cref{sec:tape}, and we show no such machines exist, even without bounding the delay:

\begin{theorem}
\label{thm:noqueue}
There exists no Hamiltonian queue machine.
\end{theorem}

The proof is by contradiction and uses a pumping-like argument: we start from the last state visited by the machine, we argue that this state cannot occur on any other word, and then successively eliminate more and more states by building longer and longer words.
Note that the result applies even without requiring machines to be constant-delay, but that it would not apply if we could have unlimited auxiliary memory: indeed it is known that
we can produce all binary words with left-push and right-pop operations. %

\begin{proof}[Proof of \cref{thm:noqueue}]
For the sake of contradiction, assume there is a Hamiltonian queue machine $\calM$ with set of states $Q$, which reads $\rho$ bits at a time. 
Let $n_0>\rho$ be some integer: we consider the word $w_0$ of size $n_0$, such that $c=(q_0, w_0)$ is the last configuration of $\calM$ before reaching the halting state.
Let us consider the words $0w_0$ and $1w_0$: they are both enumerated by $\calM$ when generating the words of size $n_0+1$, in states $q$ and $q'$ respectively. 
Remark that $q$ and $q'$ cannot both be equal to $q_0$, because whenever $\calM$ is in configuration $(q_0,w'w_0)$ for some $w' \in \Sigma^*$ then it will reach the halting state (as the transition only depends on the last $\rho$ bits $w_0$ and on the current state $q_0$). W.l.o.g., we assume that we have $q\neq q_0$. Hence, let $q_1 \coloneq q$ and $u_1 \coloneq 0w_0$. We consider the run of~$\calM$ from $(q_1,u_1)$ up to the moment when the halting state is reached, and we let $v_1$ be the word formed of all words pushed to the left of $u_1$ during this run of $\calM$ until $\calM$ halts. 
For instance, if~$\calM$ pushes~$x_0$, then~$x_1$, then~$x_2$ and then halts, then~$v_1=x_2x_1x_0$.
Note that~$v_1$ can be exponentially longer than~$u_1$.

Notice now that the run of $\calM$ on $(q_1,wv_1u_1)$, for any~$w\in\{0,1\}^*$, will use exactly the same transitions (in the same order) as the run on $(q_1,u_1)$, and then it halts.
Indeed, an easy induction shows that the leftmost $\rho$ bits will always be equal in both cases.
In particular, all these runs have the same length~$n_1$ (measured as the number of transitions performed).
Let us now consider the four words $wv_1u_1$ for $w \in \{00,01,10,11\}$, and the run of~$\calM$ for words of length~$\ell=|wv_1u_1|$.
For the same reason as before, at most one of these four words can be visited in state $q_0$ during that run.
Moreover, only one of these four words can be visited in state~$q_1$ during that run.  Indeed, only one configuration can be at position~$n_1$ from the end.
Thus, pick a word $u_2 = wv_1u_1$ with~$w\in\{0,1\}^*$ which is visited in a state $q_2 \notin \{q_0,q_1\}$ during the run of~$\calM$ from its initial configuration.
We then let~$v_2$ be the concatenation of all words pushed to the left of $u_2$ by~$\calM$ on the run from $(q_2, u_2)$ until $\calM$ halts.

This method easily generalizes to build sequences~$(u_0,u_1,\ldots ,u_k),(v_0,v_1,\ldots,v_k)$ of words and~$(q_0,q_1,\ldots,q_k)$ of states that satisfies:
for every~$i<k$, there is at most one word~$w$ of length $\left\lceil \log(k+1)\right\rceil$
such that the run of $\calM$ from its initial configuration reaches $(q_i,wv_iu_i)$.
When~$k$ equals the number of states in~$\calM$, we reach a contradiction.
\end{proof}

Note that the proof above also rules out the existence of Hamiltonian queue machines in somewhat more powerful models, e.g., if the machine is also allowed to edit in-place the $\rho$ leftmost and rightmost characters for some constant $\rho$ without changing the length.

\subparagraph{Stack machines.}
We now move from queue machines to \emph{stack machines}. Stack machines are defined like queue machines, except that they can only push and pop on the right endpoint. 

Formally, a \emph{stack machine} is a tuple $\calM = (Q, \qi, \qh,
Q_\OO, \delta)$ with 
$\qi$ and $\qh$ respectively the \emph{initial state} and \emph{halting state},
with $Q_\OO$ the \emph{output states}, and $\delta \colon (Q \setminus \{\qh\}) \times \{0,1,\bow\}
\to Q \times (\{\text{pop}\} \cup (\{\text{push}\} \times \bool))$, where $\bow$ signals that the stack is empty.
The semantics are defined as expected: based on the current state and the rightmost bit of the current word (or the $\bow$ marker),
the machine transitions to some state and either pops the rightmost bit or
pushes a bit to the right of the word. (Popping on the empty word has no effect.)

With this definition, stack machines obviously cannot achieve constant-delay enumeration of $\{0,1\}^\ell$ for arbitrary $\ell \in \NN$. Indeed, for any delay bound $B \in \NN$, the machine will not have time to change the bits at distance greater than $B/2$ from the right endpoint, which is problematic as soon as $\ell > B$. However, one can ask whether \emph{Hamiltonian} stack machines can exist without the constant-delay requirement.
The answer is negative.

\begin{proposition} %
\label{prp:missing}
Let $\calM$ be a stack machine and $k$ be its number of states. Then there are words of length $k+2$ that the machine never visits.
\end{proposition}

\label{apx:stack-lower-bound}
To prove \cref{prp:missing}, we first prove a general lemma reminiscent of pumping lemmas:

\begin{lemma}
  \label{lem:nostack}
  Let $\calM$ be a stack machine with state space $Q$, and let $k > |Q|$. Assume that the machine visits a word $u$ of length $n$ at time $s$ and later visits a word $v$ of length $n+k$ at time $t>s$, with all words between $s$ and $t$ having length $>n$. Then the machine never halts and all words visited after $t$ have $u$ as a prefix.
\end{lemma}

Intuitively, we show that we must reach on one word $u'$ on a state $q$ and reach a larger word $v'$ on the same state $q$ without having popped any bit from~$u'$. Once this has happened the sequence repeats indefinitely. The formal proof is below:

\begin{proof}[Proof of \cref{lem:nostack}]
  Fix a word $u$ of length $n$ visited at a time $s$ and $v$ a word of length $n+k$ visited at a time $t>s$, with all words between $s$ and $t$ having length $>n$. Define a sequence of timestamps $s, \ldots, s_k$ and corresponding visited
  words $u_0, \ldots, u_k$ in the following way: $s_0 \coloneq s$ and $u_0 \coloneq u$, then $s_1$ is the moment where we visit a word $u_1$ which is the last word of length $n+1$ visited between $s_0$ and $t$, and $s_2$ is the last moment where we visit a word $u_2$ which is the last word of length $n+2$ visited between $u_1$ and $v$, and so on. At the end, the time $s_k$ is such that $s_k \leq t$ and $u_k$ is a word of length $n+k$ which is the last word of length $n+k$ visited between $u_{k-1}$ and~$v$, so it must be the case that $s_k \coloneq t$ and $u_k \coloneq v$. Further note that the definition ensures that the $u_i$ are successive prefixes of one another: each $u_i$ with $0\leq i\leq k$ is a prefix of length $n+i$ of~$v$
  
  By the pigeonhole principle, there are two indices $1 \leq i < j \leq k$ such that the state reached by the machine at time $s_i$ and $s_j$ (reaching words $u_i$ and $u_j$ respectively) is the same state $q$. Recall that we have $u_i = u (v_{n+1} \cdots v_{n+i})$ and $u_j = u (v_{n+1} \cdots v_{n+j}) = u_i x$ with $x = (v_{n+i+1} \cdots v_{n+j})$.

  Now, consider the execution of the machine from time $s_j$ onwards. At $s_i$, the machine was in the same state, and it followed a sequence of transitions that never popped below depth $n+i$ until reaching $s_j$: indeed $s_i$ is the last visit of a word of length $n+i$ between $s$ and $t$ so if the machine popped below $n+i$ between $s_i$ and $s_j$ then it would visit a word of length $n+i$ again before $t$ and this would contradict the definition of $u_i$. Then, from $s_j$, the machine will replay the same sequence of transitions, in particular not popping below $n+j$, until we reach after time $s_j-s_i$ the state $q$ and the word $u_i x^2$. We can repeat this argument indefinitely to show that the machine never halts and all words produced from that point onwards start with $u_i$, hence with~$u$.
\end{proof}

We can now prove \cref{prp:missing}:

\begin{proof}[Proof of \cref{prp:missing}]
  We reason by contradiction. Let $\calM$ be a stack machine, let $k$ be the number of states of~$\calM$, let $\ell \coloneq k+2$ as in the statement, and assume by way of contradiction that $\calM$ visits all words of length $\ell$. 
  
  Consider the first word $u$ of length $k+2$ that $\calM$ visits whose first two bits $xy$ are different from $00$, and consider the moment where $\calM$ last visited the word $xy$ before visiting $u$ (this must have happened because the initial tape starts with $00$ so $\calM$ must have changed them and visited $u$). From \cref{lem:nostack} we know that all words visited afterwards start with $xy$. But, letting $x'y'$ be a sequence of two bits such that $x'y' \neq xy$ and $x'y' \neq 00$, we know that no word of length $\geq k+2$ starting with $x'y'$ was visited, so $\calM$ cannot visit all words of length $\ell$.
\end{proof}

This result also implies that prefix-Hamiltonian stack machines do not exist, or even stack machines that visit all words of length $\ell$ for sufficiently large $\ell$ (in particular they cannot output all such words, even with duplicates).
This also rules out the existence of stack machines to enumerate other languages than $\{0,1\}^\ell$, for instance $\{0,1\}^{\leq l}$ (the binary words of length up to~$\ell$), or the infinite language $\{0,1\}^*$. 

However, with unlimited memory, some languages can be enumerated in constant delay by performing push-right and pop-right operations, similarly to what was done in Appendix~\ref{sec:push-pop}.
Since we cannot enumerate $\{0,1\}^\ell$ in constant-delay, we focus on other languages.

We first note that enumeration for the language $\{0,1\}^{\leq\ell}$ can easily be done in constant-delay with push-right and pop-right operations. This contrasts with the nonexistence of stack machines for this language (as follows from \cref{prp:missing}).

\begin{claim}
  For any $\ell \geq 0$ we can design a (length-dependent) push-pop-right Gray code that produces all words over $\{0,1\}$ of length \emph{up to $\ell$} with a constant number of push-right and pop-right operations between any two produced words.
\end{claim}

\begin{proof}
This is simply by exploring the complete binary tree of depth $\ell$ and using again the even-odd trick.
\end{proof}

We last point out that the above is not possible if we want to produce the infinite sequence of the words of $\{0,1\}^*$ (i.e., of arbitrary length). This is intuitively because we must eventually produce all short words and then we can no longer change the first bits of the word. This is made formal in \cite[Theorem~6.2]{amarilli2023enumerating}.%
\section{Conclusion and Future Work}
\label{sec:ccl}

We have studied how to enumerate all binary words of length~$\ell$ in constant delay and constant auxiliary space. We defined two tape machines to perform that task ($\calT_0$ and $\calT_1$), one of which also implements a Gray code ($\calT_1$). We also constructed two deque machines that meet our goal ($\calD_1$ and $\calD_2$), all built by improving a simpler deque machine ($\calD_0$) to make it halt at the end of the enumeration. Our work opens two main lines of interesting questions.

First, can our constructions be improved?
For tape machines, we have shown that~$\calT_1$ surprisingly implements a code
similar to the 3-skew-tolerant Gray code of~\cite{sachimelfarb2025improved}.
The same paper presents a 2-skew-tolerant code, and it is not clear if it can be implemented by a tape machine.
For deque machines, there are at least two ways to improve
$\calD_1$ and~$\calD_2$: supporting decrements, which seems to boil down to detecting that the enumeration is back at the start; and minimizing the distance between consecutive outputs, for a relevant notion of distance that generalizes Hamming-1 by taking into account cyclic shifts.
In particular, we do not know if there is a deque machine that implements a
\emph{1-push-pop Gray code}, i.e., a push-pop quasi-Gray code where each word
can be obtained by one push and one pop from the previous word, even though we
know that such codes exist (see Appendix~\ref{sec:push-pop} and
\cref{prp:dequeexists}).
Another direction for improvement is to look at even weaker models: we already know that restricting the operations allowed on deque machines makes them too weak for our purposes, even with unrestricted delay.
Can we somehow weaken tape machines and still meet our goal?

Second, can our constructions be generalized? 
Indeed, while we have focused on enumerating all binary words of length~$\ell$, our
constructions would easily extend to enumerating binary words of length~$\geq
\ell$ in increasing order of length. However, it is not clear that our
work generalizes to harder cases, such as enumerating all $n$-ary words (which
may lead to different bounds~\cite{chakraborty2018space}), or all words of an
arbitrary regular language.
More generally, it would be especially interesting to study how we can use our
techniques to enumerate all walks of a directed graph with constant delay and
constant additional memory, with possible
applications to graph database query processing.
We could approach these tasks by devising variants of our constructions relying on
efficiently encoding depth-first traversals of trees, but
unfortunately the techniques used to build~$\calT_1$ from~$\calT_0$ and~$\calD_1,\calD_2$ from~$\calD_0$ (e.g., lookahead, double traversal, storing state on the tape/deque) all rely on specific properties of~$\{0,1\}^\ell$, which are not simple to extend already in the case of larger alphabets. 
The enumeration of regular languages with constant delay has also been studied in~\cite{amarilli2023enumerating}, featuring push-pop edits (reminiscent of deque machines) and push-pop-right edits (reminiscent of stack machines). Their characterization of enumerable
languages is given in a non-uniform model, but it implies a necessary condition for
languages to be enumerable in our models. A natural question is then to understand if the algorithms of~\cite{amarilli2023enumerating} can be implemented with constant additional memory, and if possible with a machine size that is not much larger than an automaton for~$L$.

\bibliography{biblio}

\end{document}